%% file: main.tex
\documentclass[11pt,letterpaper]{article}

\input{preamble}

\title{Moment Ambiguity and the Limits of Robust Stochastic Optimization}

\author{Andr\'es Cristi\thanks{EPFL, Switzerland,
  \textsf{\{andres.cristi,matteo.russo,jiechen.zhang\}@epfl.ch}.}%
  \and
  Matteo Russo\footnotemark[1]%
  \and
  Jiechen Zhang\footnotemark[1]%
}

\date{}

\begin{document}

\maketitle

\begin{abstract}
We study fundamental information-theoretic limits of robust stochastic optimization when the distribution is known only through its exact moment sequence. We develop a unified framework that produces families of distinct distributions sharing all moments yet inducing radically different optimal decisions, thereby establishing strong impossibility results for a range of decision problems under moment ambiguity.

Our binary and $N$-way moment-matching constructions separate quantiles, threshold regions, and online selection opportunities while preserving every moment. For the newsvendor problem, moment equivalence causes quantile ambiguity, inducing any fixed or randomized order quantity to fail arbitrarily badly. For revenue maximization, no deterministic or randomized posted pricing scheme can secure a nontrivial approximation relative to the full-information benchmark.

For the secretary problem with cardinal observations, the optimal robust success probability under exact moment disclosure is the classical finite-horizon value $s_m$, even for i.i.d.\ observations. This value converges to $1/e$, compared with the full-information limit of approximately $0.5801$ due to Gilbert and Mosteller (J. Am. Stat. Assoc., 1966).

We also recover and expand upon the recent impossibility result of Correa et al. (STOC, 2026) for prophet inequalities with moment knowledge. For every fixed $r\ge1$, exact moment knowledge can yield at best a $\Theta(1/\log m)$ competitive ratio, even when competing against relaxed benchmarks based on expected $r$-th order statistics or when the algorithm is allowed to select $r$ items.
\end{abstract}

\newpage

\setcounter{tocdepth}{1}
\tableofcontents

\bigskip

\input{sections/introduction}
\input{sections/preliminaries}
\input{sections/impossibility}

\input{sections/quantile}
\input{sections/revenue}
\input{sections/secretary}
\input{sections/prophet}
\input{sections/proofs}
\input{sections/discussion}

\bibliographystyle{alpha}
\bibliography{ref}

\end{document}

%% file: preamble.tex
\usepackage[utf8]{inputenc}
\usepackage[T1]{fontenc}
\usepackage{microtype}
\usepackage{amsmath}
\usepackage{mleftright}
\usepackage{amsthm}
\usepackage[numbers]{natbib}
\usepackage[hypertexnames=false]{hyperref}
\usepackage{enumitem}
\usepackage{thm-restate}
\usepackage[margin=1in]{geometry}
\makeatletter
\@ifclassloaded{acmart}{}{\usepackage{amssymb}}
\makeatother
\hypersetup{
	colorlinks,
	linkcolor={red!50!black},
	citecolor={blue!50!black},
	urlcolor={blue!80!black}
}
\usepackage{graphicx}
\usepackage{cleveref}
\usepackage{mathtools}
\usepackage{tabularx}  
\usepackage{xspace}
\usepackage[noend]{algpseudocode}
\usepackage{algorithm}
\usepackage{tikz}
\usetikzlibrary{patterns}
\usepackage{booktabs}

\usetikzlibrary{arrows.meta,positioning,calc,matrix,fit,decorations.pathreplacing,shapes.geometric,backgrounds}
\usepackage{xcolor}
\usepackage{pgfplots}
\usetikzlibrary{arrows.meta}
\pgfplotsset{compat=1.18}

\newtheorem{theorem}{Theorem}[section]

\newtheorem{lemma}{Lemma}[section]
\newtheorem{remark}{Remark}[section]
\newtheorem{corollary}{Corollary}[section]
\newtheorem{proposition}{Proposition}[section]
\newtheorem{definition}{Definition}[section]

\definecolor{figblue}{HTML}{0072B2}
\definecolor{figorange}{HTML}{D55E00}
\definecolor{figgreen}{HTML}{009E73}
\definecolor{figpurple}{HTML}{CC79A7}
\definecolor{figgray}{HTML}{666666}
\tikzset{
  paperfig/.style={font=\footnotesize,>=Latex},
  figaxis/.style={->,draw=black!70,line width=.65pt},
  figguide/.style={draw=black!25,line width=.4pt},
  figarrow/.style={->,draw=black!65,line width=.7pt},
  figbox/.style={
    draw=black!55,
    fill=white,
    rounded corners=1.5pt,
    inner xsep=5pt,
    inner ysep=3pt,
    align=center
  },
  nustyle/.style={draw=figblue,line width=.9pt},
  mustyle/.style={draw=figorange,line width=.9pt,densely dashed},
  figselected/.style={draw=figpurple,line width=1.1pt,densely dashed}
}

\newcommand{\E}{\mathbb{E}}
\newcommand{\R}{\mathbb{R}}

\newcommand{\C}{\mathbb{C}}

\def\N{\mathbb{N}}

\renewcommand{\P}{\mathbb{P}}

\newcommand{\Profit}{\mathrm{Profit}}

\newcommand{\Rev}{\mathrm{Rev}}

\newcommand{\VaR}{\operatorname{VaR}}
\newcommand{\LCVaR}{\underline{\operatorname{CVaR}}}

%% file: sections/introduction.tex
\section{Introduction}\label{sec:introduction}

Many models in online decision making, resource allocation, inventory control, and mechanism design assume that the relevant distributions are known exactly. This assumption underlies classical prophet inequalities \cite{KrengelSucheston1978,SamuelCahn1984}, newsvendor prescriptions \cite{edgeworth1888mathematical,arrow1951optimal}, and Bayesian posted pricing \cite{Myerson81}. In applications, however, a decision maker may only see summaries of a distribution: means, variances, higher moments, or other aggregate descriptors. Moment ambiguity sets are therefore a central object in distributionally robust optimization and robust mechanism design \cite{DelageY10,WiesemannKS14,rahimian2019distributionally,che2022robustly,CarrascoEtAl2018}.

We ask how well a decision maker can perform when given every moment of the relevant distributions. The benchmark is the decision that could be made with full distributional knowledge. Even this complete sequence of exact summaries can leave substantial ambiguity: two compatible distributions may have very different quantiles or useful price thresholds, and their online selection problems may offer different opportunities to stop.

To obtain lower bounds, we construct distributions that share their entire moment sequence while controlling the atoms and tails that matter to each objective. Any policy receives the same ex ante moment information on these instances. The challenge is then to show that a common decision rule cannot exploit their different opportunities, even when it can randomize or observe values online.

\subsection{Our Contributions and Technical Summary}

We construct families of moment-equivalent distributions with explicit bounds on their atoms and tails. A binary pair supplies the quantile and order-statistic separations; an $N$-way family supplies many alternative thresholds and designated values for pricing and secretary success probability.

\paragraph{Key Constructions.}
We use two separator constructions.
\begin{enumerate}[label=(\Alph*),leftmargin=*]
\item \underline{\emph{Binary separator} (\textbf{\Cref{thm:binary-separator}}).} We construct two distributions that agree on all moments and have disjoint supports by splitting a moment-annihilating sequence into its positive and negative parts. The construction is guided by Vandermonde null vectors: in finite dimension, the signs alternate (\Cref{lem:finite-binary-vandermonde}); in the infinite construction, rapidly separated nodes make the limiting weights summable and keep all moments finite. Controlled atoms and tail decay then yield the quantile and order-statistic separations in \Cref{cor:quantile-separation,cor:order-stat-separator}.
\item \underline{\emph{$N$-way residue-class separator} (\textbf{\Cref{thm:nway-separator}}).} We construct $N$ moment-equivalent distributions supported on disjoint residue classes. Roots of unity force all residue-class moment sums to agree, while the coefficient estimates in \Cref{thm:nway-separator} provide the local mass control used for randomized pricing and yield the first-atom concentration in \Cref{cor:first-atom-concentration} used by the secretary embedding. This $N$-way structure creates many incompatible decision regions under one common moment disclosure.

Our $N$-way construction gives a one-sided realization with a common nonnegative coefficient sequence and explicit bounds on its full profile. The disjoint supports in both separator constructions imply total-variation distance $1$ between distinct members.

\end{enumerate}

\paragraph{Main Results.}
These constructions give the following lower bounds.
\begin{enumerate}[label=(\arabic*),leftmargin=*]
\item \underline{\emph{Quantiles, risk, and newsvendor} (\textbf{\Cref{cor:quantile-separation} and \Cref{thm:newsvendor}}).} Exact moment information does not determine quantiles even approximately. This yields separations for Value-at-Risk, lower-tail average quantiles, chance-constrained capacity, and the newsvendor problem, including randomized order quantities; the risk and feasibility consequences are stated in \Cref{cor:risk-quantile-separation,cor:chance-capacity}.
\item \underline{\emph{Posted pricing} (\textbf{\Cref{thm:revenue-deterministic,thm:randomized-pricing-all-moments}}).} The binary separator rules out any positive uniform approximation factor for deterministic prices, and the $N$-way separator does the same for randomized prices. For the latter, disjoint good-price intervals allow a many-scale averaging argument inside one full moment class; we discuss its relation to earlier pricing lower bounds in \Cref{sec:revenue}.
\item \underline{\emph{Secretary} (\textbf{\Cref{thm:full-secretary-success}}).} In the labeled cardinal-value random-order model, the optimal robust success probability under exact moment disclosure is the classical value $s_m$ for every horizon $m$, both for independent and for i.i.d.\ values. The classical secretary rule supplies the matching lower bound.

The success-probability upper bound combines first-atom concentration with the finite Ramsey argument in \Cref{lem:ramsey-erasure}. Uniform mixtures of the moment-equivalent laws give i.i.d.\ instances, and a large homogeneous subgrid makes repeated designated values unlikely.
\end{enumerate}

\paragraph{Order-statistic prophets and relation to prior constructions.}
The binary separator also gives a self-contained upper-bound proof for the fixed-$r$ order-statistic prophet problem and its $r$-selection variant (\Cref{thm:prophet,cor:prophet-r-selection}). Together with the lower bound of Correa et al.~\cite{CorreaCristiLivanosVerdugoZhang26}, this yields the logarithmic barrier. Their moment-equivalent hard instances use a related Vandermonde mechanism and also admit an order-statistic extension. Here the atom and tail estimates make the upper bound a direct application of the same pair used for quantiles and newsvendor.

\subsection{Related Work}

\paragraph{The moment problem and Stieltjes classes.}

The classical moment problem asks when a measure is determined by its moment sequence \cite{stieltjes1894recherches,Hausdorff21-I,Hausdorff21-II,ShohatT43,Sodin19}. Chihara and Leipnik constructed discrete laws on bilateral geometric orbits that share the lognormal moment sequence \cite{Chihara70,Leipnik82,AGK19}. Suitable choices give disjoint atom sets and can also supply the designated-atom concentration and price separation used in our $N$-way applications. Our construction realizes these properties on a one-sided grid with a common nonnegative coefficient sequence and explicit bounds. Other discrete Stieltjes classes include the log-Heine constructions of Ostrovska and Turan \cite{OstrovskaTuran2018}.

\paragraph{Distributionally robust optimization.}
Distributionally robust optimization studies worst-case decisions over ambiguity sets specified by support, moments, or probability metrics \cite{scarf1957min,DelageY10,GohS10,WiesemannKS14,rahimian2019distributionally}. We compare a moment-based policy with the distribution-specific benchmark. Our lower bounds show how much this comparison can cost even when the ambiguity set fixes the full moment sequence: quantile, newsvendor, and pricing guarantees can deteriorate arbitrarily, while the online selection objectives exhibit the barriers stated above.

\paragraph{Prophet inequalities and partial information.}
Classical prophet inequalities assume known distributions \cite{KrengelSucheston1978,SamuelCahn1984}; later work studies richer feasibility constraints and posted-price interpretations \cite{HajiaghayiKleinbergSandholm2007,KleinbergWeinberg2012,RubinsteinSingla2017,LeeSingla2018,DuttingEtAl2020,FeldmanSZ21}. Under partial information, sample access and inaccurate-prior models can still yield strong guarantees \cite{CorreaDuttingFischerSchewior2022,DuttingKesselheim2019}. Correa et al.~\cite{CorreaCristiLivanosVerdugoZhang26} prove the tight $\Theta(1/\log m)$ barrier for the expected-maximum benchmark under moment information. Although their theorem is stated for the maximum, their hard instances can also be adapted to every fixed order statistic because each active block contains many copies of its designated value. \Cref{sec:prophet} gives an alternative, modular upper-bound derivation from our binary separator.

\paragraph{Newsvendor and pricing.}
The classical newsvendor solution is determined by a critical quantile
\cite{edgeworth1888mathematical,arrow1951optimal}, while posted-price revenue is determined by threshold
tail probabilities \cite{Myerson81}.  Prior work studies sample-based
newsvendor policies \cite{levi2015data}, moment-based robust newsvendor
prescriptions \cite{scarf1957min,DasDharaNatarajan21}, and robust pricing or mechanism design
under coarse information \cite{BabaioffBDS17,CarrascoEtAl2018,GiannakopoulosPocasTsigonias2022}.  The geometric
interval lower bound of Babaioff et al.~\cite{BabaioffBDS17} and the
continuum-of-scales, equal-revenue construction of Giannakopoulos,
Poças, and Tsigonias-Dimitriadis \cite{GiannakopoulosPocasTsigonias2022} are especially close to
the final averaging argument in our randomized-pricing lower bound.
At a more abstract minimax level, Hartline and Roughgarden use an
equal-revenue distribution over value scales to certify the optimal
logarithmic price lottery \cite{HartlineR14}, and Hartline and
Johnsen later formulate the corresponding point-mass/equal-revenue
decomposition through equivocal blends \cite{HartlineJ24}.
Our results isolate a different information-theoretic limitation:
distributions can require incompatible quantiles and threshold regions
even when they agree on their entire moment sequence.

\paragraph{Secretary and best-choice selection.}
For i.i.d.\ observations from a known continuous distribution, the full-information best-choice problem has asymptotic value of about $0.5801$ \cite{GilbertMosteller1966}. Esfandiari et al.~\cite{EsfandiariEtAl2020} study the random-order model with known, independent, non-identically distributed values, and Nuti \cite{Nuti2022} shows that its worst-case value coincides with the i.i.d.\ full-information benchmark. Our success-probability result gives a contrasting information boundary: replacing the distributions by their exact moment sequences reduces the robust value to the classical secretary probability.

\subsection{Organization}

The remainder of the paper is organized as follows. \Cref{sec:preliminaries} fixes notation, moment classes, and the moment-based policy model. \Cref{sec:impossibility} develops the two core moment-matching constructions; their quantitative consequences are stated alongside the applications that use them. We treat quantile ambiguity and newsvendor in \Cref{sec:quantile}, posted-price revenue in \Cref{sec:revenue}, secretary success probability in \Cref{sec:secretary}, and order-statistic prophet inequalities in \Cref{sec:prophet}. \Cref{sec:proofs} supplies the deferred proofs of the two separator constructions, followed by concluding remarks in \Cref{sec:discussion}.

%% file: sections/preliminaries.tex
\section{Preliminaries}\label{sec:preliminaries}
We start by recalling some useful notation defined for a distribution $\psi$ on $\R_+=[0,\infty)$: namely, let $F_\psi(x):=\P_\psi(X\le x)$ denote its CDF and let $F_\psi(x^-):=\P_\psi(X<x)$ be its ``open at $x$'' counterpart. Define the quantile function
\[
Q_\psi(u):=\inf\{x\ge 0:\ F_\psi(x)\ge u\},
\qquad u\in(0,1).
\]
Next, for independent random variables $X_1,\dots,X_m$, we write
$X_{(1)}\ge X_{(2)}\ge\cdots\ge X_{(m)}$ for the decreasing order statistics. 

We say that a measure is supported on a set when it assigns that set full measure. For discrete laws, we specify the sets of atoms; these need not be closed.

\subsection{Background on Moment Ambiguity}

The $k$-th moment of a distribution $\psi$ is
\[
m_k(\psi):=\E_\psi[X^k],
\]
whenever finite. Distributions agree on all moments if all these finite moments exist and agree for every $k\in\N_0$.

\begin{definition}\label{def:stieltjes-type-class}
For a moment sequence $m=(m_k)_{k\ge0}$, let
\[
\mathcal M(m):=\{\psi:\ \E_\psi[X^k]=m_k\text{ for every }k\ge0\}.
\]
We call a subfamily $\mathcal S\subseteq \mathcal M(m)$ a \emph{Stieltjes-type class} if $\mathcal S$ contains more than one distribution. In this paper we focus on structured Stieltjes-type subclasses: subfamilies whose supports, atoms, and threshold regions satisfy additional quantitative separation properties.
\end{definition}

\begin{definition}\label{def:moment-based-policy}
A policy is \emph{moment-based} if its rule is measurable with respect to the disclosed moment sequences and the observations revealed online. In particular, if two input distributions have identical moment sequences, the policy receives identical ex ante information under them.
\end{definition}

\subsection{Background on Mechanism Design and Stochastic Optimization Problems}\label{sec:defn-problems}

\paragraph{Value-at-Risk.}
Value-at-Risk is the quantile $\VaR_u(X):=Q_X(u)$. We use the lower-tail average quantile
\[
\LCVaR_u(X):=\frac{1}{u}\int_0^u Q_X(t)\,dt,
\qquad u\in(0,1).
\]
\paragraph{The Newsvendor Problem.} In the newsvendor problem, demand $D\sim\psi$ is nonnegative and the normalized profit from ordering $q\ge0$ units is
\[
\Profit_{\psi,\beta}(q):=\E_\psi[\min\{D,q\}]-(1-\beta)q,
\qquad \beta\in[0,1].
\]
A randomized ordering rule is a probability measure $\lambda$ on
$[0,\infty)$ with finite first moment. Its expected profit is
\[
\Profit_{\psi,\beta}(\lambda):=\int_{[0,\infty)}\Profit_{\psi,\beta}(q)\,d\lambda(q).
\]

\paragraph{Posted-price revenue.}
For a buyer value $X\sim\psi$, the revenue from price $v>0$ is
\[
\Rev_\psi(v):=v\P_\psi(X\ge v),
\qquad
\Rev_\psi^*:=\sup_{v>0}\Rev_\psi(v).
\]
For a randomized pricing rule $\lambda$ on $(0,\infty)$, define
\[
\Rev_\psi(\lambda):=\int_{(0,\infty)}\Rev_\psi(v)\,d\lambda(v).
\]
Randomization is unnecessary under a known distribution, but it matters under ambiguity because the same moment information may be consistent with many distributions.

\paragraph{Secretary success.}

In the labeled independent random-order cardinal-value secretary problem, the
$m$ labels have independent nonnegative values $X_1,\dots,X_m$ and arrive in a
uniformly random order, independently of those values. At each arrival, the
policy observes the item's label and realized value and must either reject it
irrevocably or stop and select it; at most one item may be selected. Success
means selecting any maximizer. Under exact moment disclosure, the policy is
given, for each label $i$, the full sequence
$(\E[X_i^k])_{k\ge0}$ before the arrivals and no other ex ante distributional
information. A moment-based secretary policy is a policy in the sense of
\Cref{def:moment-based-policy} for this disclosure and observation process.
Its robust success guarantee is the infimum of its success probability over
all independent per-label distributions with finite moments of every order,
where the policy observes the resulting moment sequence of each label. Thus
the worst case ranges over both the disclosed moment profiles and the
compatible distributions. The optimal robust value is the supremum of this
guarantee over moment-based secretary policies.

We also consider the restriction to i.i.d.\ values. In this model, the policy knows that the values are identically distributed and receives the full moment sequence of their common distribution; the robust guarantee ranges over all such disclosures and compatible common distributions.

For each integer horizon $m\ge1$, the finite-horizon classical secretary value is
\[
s_m:=
\max_{r\in\{1,\dots,m\}}
\frac{r-1}{m}\sum_{t=r}^{m}\frac{1}{t-1},
\]
where the expression for $r=1$ is interpreted as $1/m$. The classical
asymptotic is $\lim_{m\to\infty}s_m=1/e$.

\paragraph{Prophet inequalities.}
In the prophet inequality problem, values $X_1,\dots,X_m$ are revealed sequentially and the decision maker chooses a stopping time $\tau$. We allow the rule to make no selection, in which case its reward is zero. The classical benchmark is $\E[\max_i X_i]$. In \Cref{sec:prophet}, we also compare against the weaker benchmark $\E[X_{(r)}]$, the expected $r$-th largest value, for fixed $r$.

%% file: sections/impossibility.tex
\section{A Residue-Class Moment-Matching Framework}\label{sec:impossibility}

This section develops two constructions for producing distinct distributions that agree on all moments. The binary construction splits an alternating signed identity into two probability measures; the $N$-way construction uses roots of unity to make the moment sums of $N$ residue classes coincide. The binary pair is enough for two-instance separations, whereas the $N$-way family is needed when a randomized policy may spread its actions across many disjoint decision regions.

Finite-dimensional cancellation is only the algebraic starting point. On an infinite support, the weights must be summable and admit a common positive normalization, every distribution must have finite moments of every order, and the cancellation identities must remain valid after passing to the limit. Our applications require still more: quantitative control of designated atoms and of the remaining tail. We state these properties here as a toolkit and defer the analytic estimates to \Cref{sec:proofs}.

\subsection{The Binary Construction}

We first isolate the finite-dimensional mechanism. On $k$ prescribed nodes, the binary construction is the one-dimensional nullspace of a rectangular Vandermonde matrix.

\begin{lemma}\label{lem:finite-binary-vandermonde}
Let $k\ge2$ and let $x_0<x_1<\cdots<x_{k-1}$ be distinct real numbers. There exist two distinct probability distributions $\mu_+$ and $\mu_-$, supported on disjoint subsets of $\{x_0,\dots,x_{k-1}\}$, such that
\[
\int x^s\,d\mu_+(x)=\int x^s\,d\mu_-(x),
\qquad s=0,1,\dots,k-2.
\]
\end{lemma}

\begin{proof}
Define the barycentric weights
\[
a_i=\frac{1}{\prod_{\ell\ne i}(x_i-x_\ell)},
\qquad i=0,\dots,k-1.
\]
Then
\[
\sum_{i=0}^{k-1} a_i x_i^s=0,
\qquad s=0,1,\dots,k-2.
\]
Indeed, the vector $(a_i)_i$ lies in the nullspace of the $(k-1)\times k$ Vandermonde matrix with rows indexed by $s=0,\dots,k-2$; since this matrix has rank $k-1$, the nullspace is one-dimensional. Equivalently, the identities follow from the leading-coefficient identity in Lagrange interpolation of every polynomial of degree at most $k-2$. Since the nodes are strictly increasing, the signs of the $a_i$ alternate. Thus the positive and negative parts are both nonempty, and the equation for $s=0$ gives equal total mass.

Let
\[
P=\{i:a_i>0\},\qquad M=\{i:a_i<0\},
\qquad
W=\sum_{i\in P}a_i=-\sum_{i\in M}a_i,
\]
and define
\[
\mu_+(\{x_i\})=
\begin{cases}
a_i/W, & i\in P,\\
0, & i\notin P,
\end{cases}
\qquad
\mu_-(\{x_i\})=
\begin{cases}
-a_i/W, & i\in M,\\
0, & i\notin M.
\end{cases}
\]
These are probability distributions on disjoint supports. Dividing the identities for $\sum_i a_i x_i^s$ by $W$ gives the claimed moment matching through degree $k-2$.
\end{proof}

\begin{remark}\label{rem:finite-vandermonde-tightness}
If two probability distributions supported on the same $k$ fixed points agree on moments $s=0,1,\dots,k-1$, then they are identical. Their difference vector lies in the kernel of the full $k\times k$ Vandermonde matrix, which is invertible. Thus matching through degree $k-2$ is the strongest possible nontrivial finite matching on $k$ fixed support points.
\end{remark}

To pass from finitely many equations to all moments, we choose a rapidly separated support and take a controlled infinite-support limit of the Vandermonde construction. The decay estimate in the next theorem simultaneously guarantees normalization, finiteness of every moment, and the local mass bounds used later.

\begin{theorem}\label{prop:key-prop-1}\label{thm:binary-separator}
Fix integers $n,B\ge 3$ and define
\[
x_0=1,\qquad x_1=B,\qquad x_j=(nB)^{B^{j-2}}\quad (j\ge 2).
\]
There are weights $(w_j)_{j\ge0}$ and a normalizing constant $W>0$ such that
\[
\nu(\{x_j\}) =
\begin{cases}
w_j/W, & j\text{ even},\\
0, & j\text{ odd},
\end{cases}
\qquad
\mu(\{x_j\}) =
\begin{cases}
|w_j|/W, & j\text{ odd},\\
0, & j\text{ even},
\end{cases}
\]
are probability distributions on disjoint supports. Both distributions have finite moments of every order and agree on all moments.
Writing $p_j:=\nu(\{x_j\})$ and $q_j:=\mu(\{x_j\})$, the following quantitative bounds hold:
\[
p_0\ge 1-\frac{2}{n},\qquad
q_1\ge 1-\frac{1}{n^2B^2},\qquad
p_2\ge \frac{1}{4n},
\]
with $q_0=p_1=q_2=0$. Moreover, $W\ge n$ and
\[
|w_j|\le \frac{8(nB)^2}{x_j^2}\qquad\text{for all }j\ge0.
\]
\end{theorem}

The first three bounds expose the local structure: $\nu$ places almost all of its mass at $1$ and a $1/\Theta(n)$ mass at $nB$, whereas $\mu$ places almost all of its mass at $B$. The remaining tail is small enough for the quantile and order-statistic applications.
\Cref{fig:binary-separator-support} summarizes the alternating support pattern and the designated atoms used later.

\input{figures/binary_separator}

\subsection{\texorpdfstring{$N$}{N}-Way Residue-Class Construction}

A binary pair suffices when a lower bound compares two regimes. To control randomized policies that may distribute their actions over many disjoint regions, we need arbitrarily many distributions with one common moment sequence. The next construction replaces the binary sign split by residue classes
and uses roots of unity to force all residue-class moment sums to agree (see \Cref{fig:nway-residue-classes} for an illustration).

\begin{theorem}\label{prop:key-prop-2}\label{thm:nway-separator}
Fix integers $B,N\ge2$ and let $(w_j)_{j\in\N_0}$ be the unique coefficients that satisfy
\[
G_N(x):=\prod_{m=0}^{\infty}\sum_{u=0}^{N-1}(xB^{-m})^u
=\sum_{j=0}^{\infty}w_jx^j.
\]
There is a common normalizing constant $W\in[1,4)$ such that, for every $h\in\{0,\dots,N-1\}$,
\[
\mu_h(\{B^j\}) :=
\begin{cases}
w_j/W, & j\equiv h\pmod N,\\
0, & \text{otherwise},
\end{cases}
\]
is a probability distribution. The distributions $\mu_0,\dots,\mu_{N-1}$ are supported on disjoint residue classes, have finite moments of every order, and agree on all moments.

Moreover, given $j\in\N_0$, let $L= \lfloor j/(N-1) \rfloor$, and $b=j-L(N-1)$, and define
\[
e(j):=\frac{(N-1)L(L-1)}{2}+Lb,
\]
the coefficients $w_j$ satisfy
\[
B^{-e(j)}\le w_j\le 2^jB^{-e(j)}.
\]

In addition, $j-e(j)\le N-1$, with equality if and only if $j\in\{N-1,\dots,2N-2\}$.

\end{theorem}

\input{figures/nway_separator}

%% file: figures/binary_separator.tex
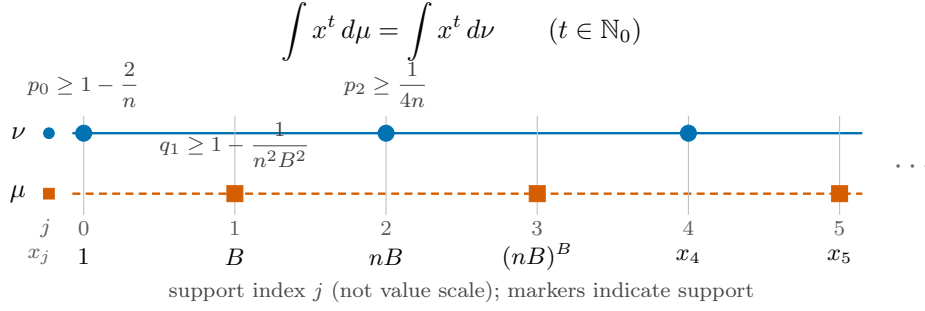
\begin{figure}[t]
\centering
\begin{tikzpicture}[
    paperfig,
    nuatom/.style={circle,draw=figblue,fill=figblue,minimum size=6pt,inner sep=0pt},
    muatom/.style={rectangle,draw=figorange,fill=figorange,minimum size=6pt,inner sep=0pt},
    bound/.style={font=\scriptsize,align=center,text=black!75}
]
    \node[font=\small] at (5.0,2.55)
        {$\displaystyle \int x^t\,d\mu=\int x^t\,d\nu\qquad(t\in\N_0)$};

    \draw[nustyle] (-0.15,1.2) -- (10.3,1.2);
    \draw[mustyle] (-0.15,0.4) -- (10.3,0.4);
    \node[nuatom,minimum size=4pt] at (-0.46,1.2) {};
    \node[left,font=\small\bfseries] at (-0.63,1.2) {$\nu$};
    \node[muatom,minimum size=4pt] at (-0.46,0.4) {};
    \node[left,font=\small\bfseries] at (-0.63,0.4) {$\mu$};

    \foreach \x in {0,2,4,6,8,10}{
        \draw[figguide] (\x,0.12) -- (\x,1.43);
    }

    \node[nuatom] (nu0) at (0,1.2) {};
    \node[nuatom] (nu2) at (4,1.2) {};
    \node[nuatom] (nu4) at (8,1.2) {};
    \node[muatom] (mu1) at (2,0.4) {};
    \node[muatom] (mu3) at (6,0.4) {};
    \node[muatom] (mu5) at (10,0.4) {};

    \node[bound,above=3pt of nu0] {$p_0\ge1-\dfrac2n$};
    \node[bound,above=3pt of mu1] {$q_1\ge1-\dfrac1{n^2B^2}$};
    \node[bound,above=3pt of nu2] {$p_2\ge\dfrac1{4n}$};

    \node[left,font=\scriptsize,text=black!65] at (-0.3,-0.05) {$j$};
    \node[left,font=\scriptsize,text=black!65] at (-0.3,-0.43) {$x_j$};
    \foreach \x/\j/\val in {
        0/0/1,
        2/1/B,
        4/2/nB,
        6/3/(nB)^B,
        8/4/x_4,
        10/5/x_5
    }{
        \node[font=\scriptsize,text=black!65] at (\x,-0.05) {$\j$};
        \node at (\x,-0.43) {$\val$};
    }

    \node[font=\large,text=figgray] at (11.0,0.75) {$\cdots$};
    \node[font=\scriptsize,text=black!72] at (5,-0.9)
        {support index $j$ (not value scale); markers indicate support};
\end{tikzpicture}
\caption{Binary separator: alternating supports and three controlled atoms.}
\label{fig:binary-separator-support}
\end{figure}

%% file: figures/nway_separator.tex
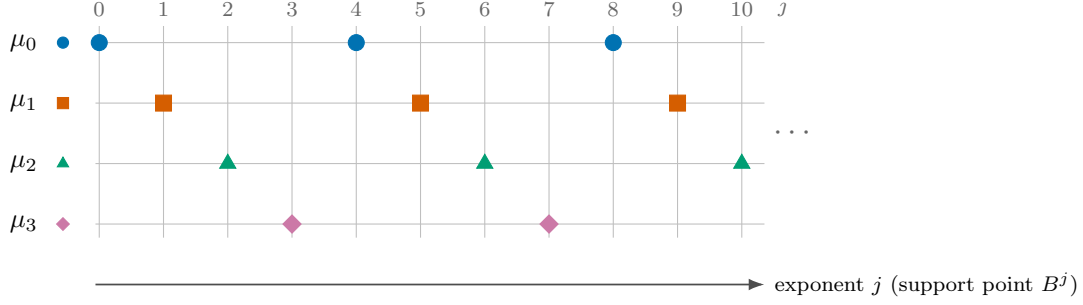
\begin{figure}[t]
\centering
\begin{tikzpicture}[
    paperfig,
    czeromark/.style={circle,draw=figblue,fill=figblue,minimum size=6pt,inner sep=0pt},
    conemark/.style={rectangle,draw=figorange,fill=figorange,minimum size=6pt,inner sep=0pt},
    ctwomark/.style={regular polygon,regular polygon sides=3,draw=figgreen,fill=figgreen,minimum size=7pt,inner sep=0pt},
    cthreemark/.style={diamond,draw=figpurple,fill=figpurple,minimum size=7pt,inner sep=0pt}
]

    \foreach \y/\h/\style in {2.4/0/czeromark,1.6/1/conemark,0.8/2/ctwomark,0/3/cthreemark}{
        \draw[figguide] (-0.05,\y) -- (8.8,\y);
        \node[\style,scale=.7] at (-0.48,\y) {};
        \node[left,font=\small\bfseries] at (-0.68,\y) {$\mu_{\h}$};
    }
    \foreach \j in {0,...,10}{
        \pgfmathsetmacro{\x}{0.85*\j}
        \draw[figguide] (\x,-0.18) -- (\x,2.62);
        \node[above,font=\scriptsize,text=black!60] at (\x,2.62) {$\j$};
    }
    \node[above,font=\scriptsize,text=black!60] at (9.05,2.62) {$j$};

    \foreach \j in {0,4,8}{
        \pgfmathsetmacro{\x}{0.85*\j}
        \node[czeromark] at (\x,2.4) {};
    }
    \foreach \j in {1,5,9}{
        \pgfmathsetmacro{\x}{0.85*\j}
        \node[conemark] at (\x,1.6) {};
    }
    \foreach \j in {2,6,10}{
        \pgfmathsetmacro{\x}{0.85*\j}
        \node[ctwomark] at (\x,0.8) {};
    }
    \foreach \j in {3,7}{
        \pgfmathsetmacro{\x}{0.85*\j}
        \node[cthreemark] at (\x,0) {};
    }
    \node[font=\large,text=figgray] at (9.2,1.2) {$\cdots$};

    \draw[figaxis] (-0.05,-0.8) -- (8.8,-0.8)
        node[right,font=\scriptsize] {exponent $j$ (support point $B^j$)};
\end{tikzpicture}
\caption{$N$-way separator for $N=4$.}
\label{fig:nway-residue-classes}
\end{figure}

%% file: sections/quantile.tex
\section{Quantile Ambiguity: VaR, Chance Constraints, and Newsvendor}\label{sec:quantile}

The binary separator shows that exact moments do not determine quantiles, even up to an arbitrary multiplicative factor. We first state this separation and its consequences for risk and robust capacity, then use the same hidden critical quantile to prove the newsvendor lower bound.

\begin{corollary}\label{cor:quantile-separation}
For every $\beta\in[1/2,1)$ and every integer $B\ge3$, there exist two distributions $\mu,\nu$ on $\R_+$ that agree on all moments and satisfy
\[
Q_\mu(\beta)=B,\qquad Q_\nu(\beta)=1.
\]
\end{corollary}

\begin{proof}
Choose an integer $n>\max\{4,2/(1-\beta)\}$ and apply \Cref{thm:binary-separator}. Since $\beta\ge1/2$, the smallest atoms of $\mu$ and $\nu$ have enough mass to contain level $\beta$: $q_1>\beta$ and $p_0>\beta$. Thus $Q_\mu(\beta)=B$ and $Q_\nu(\beta)=1$.
\end{proof}

\begin{corollary}\label{cor:risk-quantile-separation}
Fix integers $B\ge3$ and $n\ge4$, and let $\mu,\nu$ be the binary separator distributions with parameters $B,n$. Then, for every $u\in[1/2,1-2/n]$,
\[
\VaR_u(\mu)=B,\qquad \VaR_u(\nu)=1.
\]
Moreover, for every $u\in(0,1-2/n]$,
\[
\LCVaR_u(\mu)=B,\qquad \LCVaR_u(\nu)=1.
\]
\end{corollary}

\begin{proof}
Apply \Cref{thm:binary-separator}. The smallest support point of $\mu$ is $B$ and has mass at least $1-1/(n^2B^2)>1-2/n$; the smallest support point of $\nu$ is $1$ and has mass at least $1-2/n$. Hence the displayed quantiles are constant over the stated ranges, and integrating the constant quantile functions gives the lower-tail average quantiles.
\end{proof}

\begin{corollary}\label{cor:chance-capacity}
Fix $\beta\in[1/2,1)$ and $B\ge3$. There are moment-equivalent distributions $\mu,\nu$ such that any moment-based capacity rule that is feasible for every distribution in this moment class at service level $\beta$ must choose capacity at least $B$. Consequently, on distribution $\nu$ its capacity is a factor at least $B$ larger than the distribution-specific minimum feasible capacity.
\end{corollary}

\begin{proof}
Use the pair from \Cref{cor:quantile-separation}. Feasibility at service level $\beta$ means choosing $q$ with $F(q)\ge\beta$. Under $\mu$, this requires $q\ge Q_\mu(\beta)=B$. Under $\nu$, the distribution-specific minimum feasible capacity is $Q_\nu(\beta)=1$.
\end{proof}

This corollary concerns rules that require feasibility uniformly over the moment class. It shows that robust feasibility alone can force an arbitrarily large over-capacity relative to the distribution-specific minimum feasible capacity.

\subsection{Newsvendor Problem}

Recall that in the newsvendor problem, we define the expected profit of ordering $q$ units as

\[
\Profit_{\psi,\beta}(q):=\E_\psi[\min\{D,q\}]-(1-\beta)q,
\qquad \beta\in[0,1],
\]
where $D\sim \psi$ is the demand, and $(1-\beta)$ is the per-unit cost.

\begin{lemma}\label{lem:newsvendor-critical-fractile}
Fix $\beta\in[0,1]$ and let $\psi$ be a distribution on $\R_+$ with finite first moment. Then $q\mapsto\Profit_{\psi,\beta}(q)$ is concave on $[0,\infty)$. For $q>0$, its one-sided derivatives are
\[
\partial_-\Profit_{\psi,\beta}(q)=\beta-F_\psi(q^-),
\qquad
\partial_+\Profit_{\psi,\beta}(q)=\beta-F_\psi(q).
\]
The right-derivative formula also holds at $q=0$.
Consequently,
\[
\arg\max_{q\ge0}\Profit_{\psi,\beta}(q)
=
\{q\ge0:\ F_\psi(q^-)\le\beta\le F_\psi(q)\}.
\]
\end{lemma}

\begin{proof}
The tail-integral formula gives
\[
\Profit_{\psi,\beta}(q)=\int_0^\infty \P_\psi(\min\{D,q\}\ge t)\,dt-(1-\beta)q = \int_0^q\P_\psi(D\ge t)\,dt-(1-\beta)q,
\]
since $\P_\psi(\min\{D,q\}\geq t) = \P_\psi(D\ge t) \mathbf{1}_{\{q\geq t\}}$.
 The integrand in the display above is nonincreasing, so the profit is concave. The one-sided derivatives are the left and right limits of the integrand minus $1-\beta$, which gives the displayed formulas and the standard first-order optimality condition for a concave function.
\end{proof}

The following lemma states that the highest profit is achieved by a deterministic order amount, and is an immediate consequence of \Cref{lem:newsvendor-critical-fractile} together with Jensen's inequality.

\begin{lemma}\label{lem:newsvendor-randomization}
Fix $\beta\in[0,1]$, let $\psi$ be a distribution on $\R_+$ with
finite first moment, and let $\lambda$ be a randomized ordering rule. Writing
$\bar q:=\int_{[0,\infty)}q\,d\lambda(q)$, we have
\[
\Profit_{\psi,\beta}(\lambda)\le \Profit_{\psi,\beta}(\bar q).
\]
\end{lemma}

\begin{proof}

By definition, $\lambda$ has finite first moment, so $\bar q<\infty$.
Moreover,
\[
\bigl|\Profit_{\psi,\beta}(q)\bigr|
\le \E_\psi[D]+q,
\]
and hence $q\mapsto\Profit_{\psi,\beta}(q)$ is integrable under $\lambda$.
The function is concave by \Cref{lem:newsvendor-critical-fractile}, so Jensen's
inequality gives
\[
\Profit_{\psi,\beta}(\lambda)
=\int_{[0,\infty)}\Profit_{\psi,\beta}(q)\,d\lambda(q)
\le
\Profit_{\psi,\beta}\!\left(\int_{[0,\infty)}q\,d\lambda(q)\right)
=\Profit_{\psi,\beta}(\bar q).\qedhere
\]

\end{proof}

Finally, for the two moment-agreeing distributions $\mu,\nu$ constructed in \Cref{thm:binary-separator}, the lemmas shown earlier yield that both optimal order amounts and optimal profits differ by an arbitrary (multiplicative) factor.

\begin{lemma}\label{lem:newsvendor-optima}
Fix $\beta\in[1/2,1)$ and $\varepsilon\in(0,1)$. Let $B,n$ be integers satisfying
\[
B\ge \frac{4}{\varepsilon(1-\beta)},
\qquad
n>\max\left\{2B,\frac{2}{1-\beta}\right\},
\]
    and let $\mu,\nu$ be the binary separator distributions with parameters $B,n$. Then the unique optimal order quantities and profits are
\[
q_{\mu,\beta}^*=B,\qquad \Profit_{\mu,\beta}^*=\beta B,
\qquad
q_{\nu,\beta}^*=1,\qquad \Profit_{\nu,\beta}^*=\beta.
\]
\end{lemma}

\begin{proof}
Since $n>2/(1-\beta)$, we have
$\beta<1-2/n$, so \Cref{cor:risk-quantile-separation} gives
$Q_\mu(\beta)=B$ and $Q_\nu(\beta)=1$. These quantities are optimal by
\Cref{lem:newsvendor-critical-fractile}.

For uniqueness, $\Profit_{\mu,\beta}(q)=\beta q$ on $[0,B]$ and
\[
\partial_+\Profit_{\mu,\beta}(B)=\beta-F_\mu(B)=\beta-q_1<0,
\]
using $q_1> \beta$. Concavity then makes the profit strictly decreasing after $B$. The proof for $\nu$ is identical, using $p_0>\beta$ at the atom $1$. The profit values follow because the supports are contained in $[B,\infty)$ and $[1,\infty)$, respectively.
\end{proof}

For a randomized ordering rule, \Cref{lem:newsvendor-randomization} upper-bounds profit under either demand law by the profit of the deterministic order quantity at its mean $\bar q$. \Cref{fig:newsvendor-incompatible-orders} visualizes the resulting split at $\bar q=\varepsilon B$.

\input{figures/newsvendor_incompatibility}

\begin{theorem}\label{thm:newsvendor}

Fix $\beta\in[1/2,1)$ and $\varepsilon\in(0,1)$. Let $B,n$ be integers satisfying
\[
B\ge \frac{4}{\varepsilon(1-\beta)},
\qquad
n>\max\left\{2B,\frac{2}{1-\beta}\right\},
\]
and let $\mu,\nu$ be the binary separator distributions with parameters $B,n$.
Then every randomized ordering rule $\lambda$ satisfies

\[
\min\left\{
\frac{\Profit_{\mu,\beta}(\lambda)}{\Profit_{\mu,\beta}^*},
\frac{\Profit_{\nu,\beta}(\lambda)}{\Profit_{\nu,\beta}^*}
\right\}
<\varepsilon.
\]
\end{theorem}

\begin{proof}
Let $\bar q=\int_{[0,\infty)}q\,d\lambda(q)$, which is finite by the
definition of a randomized ordering rule. \Cref{lem:newsvendor-randomization} gives
\[
\Profit_{\mu,\beta}(\lambda)\le \Profit_{\mu,\beta}(\bar q),
\qquad
\Profit_{\nu,\beta}(\lambda)\le \Profit_{\nu,\beta}(\bar q).
\]
If $\bar q<\varepsilon B$, then $\bar q<B$ and
\[
\frac{\Profit_{\mu,\beta}(\lambda)}{\Profit_{\mu,\beta}^*}
\le \frac{\beta\bar q}{\beta B}<\varepsilon.
\]
If $\bar q\ge\varepsilon B$, then $\varepsilon B>1$ and the $\nu$-profit is decreasing on $[1,\infty)$. Since $\varepsilon B<B<nB=x_2$,
\[
\Profit_{\nu,\beta}(\lambda)
\le p_0+(1-p_0)\varepsilon B-(1-\beta)\varepsilon B
\le 1+\frac{2\varepsilon B}{n}-(1-\beta)\varepsilon B<0,
\]
by $1-p_0\le2/n$, $n>2B$, and $B\ge4/(\varepsilon(1-\beta))$. Thus the $\nu$ ratio is below $\varepsilon$.
\end{proof}

The binary separator therefore yields lower bounds for decisions governed by one hidden critical quantile. Posted pricing requires a stronger form of ambiguity: the next section uses the $N$-way construction to hide many disjoint threshold regions simultaneously.

%% file: figures/newsvendor_incompatibility.tex
\begin{figure}[H]
\centering
\begin{tikzpicture}[
    paperfig,
    curve/.style={line width=1.2pt}
]
    \def\xone{0.85}
    \def\xeps{3.15}
    \def\xmax{8.25}
    \def\yone{2.15}
    \def\yeps{0.82}

    \fill[figorange!5] (0.05,-0.62) rectangle (\xeps,2.75);
    \fill[figblue!5] (\xeps,-0.62) rectangle (\xmax,2.75);

    \draw[figaxis] (0,-0.02) -- (8.65,-0.02)
        node[right] {$z=\bar q/B$};
    \draw[figaxis] (0,-0.58) -- (0,2.88)
        node[above,align=center] {normalized\\profit};

    \draw[figguide,densely dashed] (\xone,-0.08) -- (\xone,\yone);
    \draw[figguide,densely dashed] (\xeps,-0.58) -- (\xeps,2.7);
    \draw[figguide,densely dashed] (\xmax,-0.08) -- (\xmax,\yone);
    \draw[figguide,densely dashed] (0,\yeps) -- (\xeps,\yeps);
    \draw[figguide,densely dashed] (0,\yone) -- (\xmax,\yone);

    \draw[curve,figorange] (0,0) -- (\xmax,\yone);
    \draw[curve,figblue,dashed] (0,0) -- (\xone,\yone) -- (\xeps,-0.36);
    \draw[curve,figblue,dashed,->] (\xeps,-0.36) -- (3.55,-0.6);

    \fill[figblue] (\xone,\yone) circle (2.5pt);
    \fill[figorange] (\xmax,\yone) rectangle +(5pt,5pt);
    \fill[figorange] (\xeps,\yeps) circle (2.1pt);
    \fill[figblue] (\xeps,-0.36) circle (2.1pt);

    \node[left] at (-0.08,0) {$0$};
    \node[left] at (-0.08,\yeps) {$\varepsilon$};
    \node[left] at (-0.08,\yone) {$1$};
    \node[below,align=center] at (\xone,-0.08)
        {$1/B$\\[-1pt]\scriptsize $q_{\nu,\beta}^*=1$};
    \node[below] at (\xeps,-0.08) {$\varepsilon$};
    \node[below,align=center] at (\xmax,-0.08)
        {$1$\\[-1pt]\scriptsize $q_{\mu,\beta}^*=B$};

    \node[fill=white,inner sep=1.5pt] at (6.15,1.55) {$\mu$};
    \node[fill=white,inner sep=1.5pt] at (1.72,1.25) {$\nu$};

    \node[align=center] at (1.7,2.58)
        {$\bar q<\varepsilon B$\\[-1pt]$\mu\text{-ratio}<\varepsilon$};
    \node[align=center] at (5.7,2.58)
        {$\bar q\ge\varepsilon B$\\[-1pt]$\nu\text{-profit}<0$};

\end{tikzpicture}
  \caption{\footnotesize Newsvendor: the two moment-equivalent demand laws require incompatible mean order quantities.}
\label{fig:newsvendor-incompatible-orders}
\end{figure}
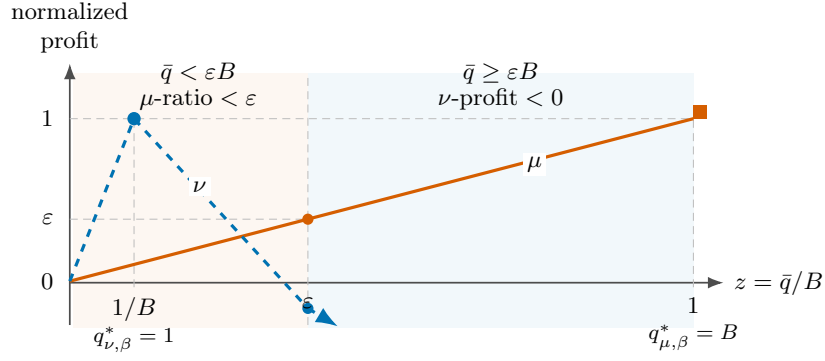

%% file: sections/revenue.tex
\section{Threshold Ambiguity: Revenue Maximization and Randomized Pricing}\label{sec:revenue}

In this section, we show that the constructions in \Cref{sec:impossibility} yield strong impossibility results for revenue maximization. To be more precise, posted-price revenue is a threshold objective since a price performs well only if it lies in the right tail region of the value distribution. The constructions in \Cref{sec:impossibility} then create moment-equivalent distributions whose good threshold regions are far apart.

\subsection{Deterministic Posted Prices}

\begin{theorem}\label{thm:revenue-deterministic}
For every $\varepsilon\in(0,1)$, there exist two distributions $\mu,\nu$ that agree on all moments such that every deterministic price $v>0$ satisfies
\[
\min\left\{
\frac{\Rev_\mu(v)}{\Rev_\mu^*},
\frac{\Rev_\nu(v)}{\Rev_\nu^*}
\right\}
<\varepsilon.
\]
\end{theorem}

\begin{proof}
Choose integers $B,n$ with $B\ge\lceil4/\varepsilon\rceil+1$ and $n\ge2B$, and let $\mu,\nu$ be the binary separator distributions. Then
\[
\Rev_\mu^*\ge \Rev_\mu(B)=B,
\qquad
\Rev_\nu^*\ge \Rev_\nu(nB)=nB\,p_2\ge B/4.
\]
Fix a price $v>0$. If $v\le1$, then $\Rev_\nu(v)\le1$, so $\Rev_\nu(v)/\Rev_\nu^*<\varepsilon$. If $1<v\le B$, the support of $\nu$ has no atom in $(1,nB)$ and
\[
\Rev_\nu(v)=v(1-p_0)\le \frac{2B}{n},
\]
so $\Rev_\nu(v)/\Rev_\nu^*\le8/n<\varepsilon$. If $v>B$, then the atom $B$ under $\mu$ does not buy, and
\[
\Rev_\mu(v)\le
\sum_{\substack{j\ge3\\ j\text{ odd}}}x_jq_j
\le
\frac1W\sum_{j\ge3}x_j|w_j|
\le
\frac{16nB^2}{x_3}.
\]
Since $x_3=(nB)^B$, this is an $\varepsilon$ fraction of $\Rev_\mu^*\ge B$ for the chosen parameters.
\end{proof}

\begin{remark}
The binary construction only rules out a single deterministic price. A lottery that mixes the two distribution-specific good prices can obtain a constant factor for this pair because revenue is nonnegative. The $N$-way construction below therefore addresses randomized rules by creating many disjoint good price regions with identical moment disclosures.
\end{remark}

\subsection{Randomized Posted Prices}

For integers $N,B\ge2$, let $n_h$ be the unique integer in $\{N-1,\dots,2N-2\}$ with $n_h\equiv h\pmod N$, and define
\[
I_h:=(B^{n_h-1},B^{n_h}].
\]
Also set
\[
C_N:=\sum_{j\notin\{N-1,\dots,2N-2\}}2^{2j-e(j)-(N-2)},
\]
where $e(j)$ is defined in \Cref{thm:nway-separator}.
This constant is finite; we establish convergence in the proof of the following localization bound.

\begin{lemma}\label{lem:N-class-revenue-localization}
Fix integers $N,B\ge2$, and let $\mu_0,\dots,\mu_{N-1}$ be the distributions from \Cref{thm:nway-separator}. For every $h\in\{0,\dots,N-1\}$:
\[
\Rev_{\mu_h}^*\ge \frac{B^{N-1}}{4},
\]
and every $v\notin I_h$ satisfies
\[
\Rev_{\mu_h}(v)\le (2^{2N-2}+C_N)B^{N-2}.
\]
\end{lemma}

\begin{proof}

First verify that $C_N$ is finite. Write $j=L(N-1)+b$, with
$L\in\N_0$ and $b\in\{0,\dots,N-2\}$. The exponent in its defining
series is $A_L+(2-L)b$, where
$A_L:=(N-1)(5L-L^2)/2-(N-2)$. For $L\ge3$, the $L$-th block is
bounded by $(N-1)2^{A_L}$, and
$A_{L+1}-A_L=(N-1)(2-L)\le-1$. These bounds form a geometric tail:
\[
\sum_{L=3}^{\infty}\sum_{b=0}^{N-2}2^{A_L+(2-L)b}
\le (N-1)2^{A_3}\sum_{k=0}^{\infty}2^{-k}<\infty.
\]
The finitely many terms with $L\le2$ cause no difficulty, so $C_N<\infty$.

For the lower bound, \Cref{thm:nway-separator} gives $w_{n_h}\ge B^{-e(n_h)}$ and $n_h-e(n_h)=N-1$, while $W<4$. Hence
\[
\Rev_{\mu_h}^*
\ge B^{n_h}\mu_h(\{B^{n_h}\})
\ge \frac{B^{N-1}}{4}.
\]
For the upper bound, first sum the contribution of all exponents outside the special block. By \Cref{thm:nway-separator},
\[
\sum_{j\notin\{N-1,\dots,2N-2\}} B^j\mu_h(\{B^j\})
\le
C_NB^{N-2}.
\]
Now fix $v\notin I_h$. If $v\le B^{N-2}$, then $\Rev_{\mu_h}(v)\le B^{N-2}$. If $B^{N-2}<v\le B^{n_h-1}$, the only special-block atom that can contribute is $B^{n_h}$, and
\[
v\mu_h(\{B^{n_h}\})
\le B^{n_h-1}w_{n_h}
\le 2^{2N-2}B^{N-2}.
\]
If $v>B^{n_h}$, no special-block atom contributes. Combining these cases gives the claim.
\end{proof}

\Cref{fig:pricing-localization} illustrates how the disjoint intervals $I_h$ convert revenue localization into a lower bound against every randomized price.
\input{figures/pricing_localization}

We now prove this section's main result:

\begin{theorem}\label{thm:randomized-pricing-all-moments}

For every $\varepsilon>0$, there exist an integer $N\ge2$ and
moment-equivalent distributions $\mu_0,\dots,\mu_{N-1}$ such that, for every
randomized pricing strategy $\lambda$, there exists an index
$h\in\{0,\dots,N-1\}$ satisfying
\[
\frac{\Rev_{\mu_h}(\lambda)}{\Rev_{\mu_h}^*}<\varepsilon.
\]

\end{theorem}

\begin{proof}
Choose an integer $N\ge2$ so that $1/N<\varepsilon/2$, and then choose an integer $B\ge2$ so large that $4(2^{2N-2}+C_N)/B \le \varepsilon/2$. Let $\mu_0,\dots,\mu_{N-1}$ be the distributions from \Cref{thm:nway-separator} with these parameters $N$ and $B$. Consider now any randomized pricing strategy $\lambda$. Since the intervals $I_0,\dots,I_{N-1}$ are pairwise disjoint, $\sum_{h=0}^{N-1}\lambda(I_h)\le 1$, and so there exists some $h_*\in\{0,\dots,N-1\}$ such that $\lambda(I_{h_*})\le 1/N$. We next split the expected revenue of $\lambda$ against $\mu_{h_*}$ into prices inside and outside $I_{h_*}$:
\begin{align*}
    \Rev_{\mu_{h_*}}(\lambda) &= \int_{I_{h_*}} \Rev_{\mu_{h_*}}(v) d\lambda(v) + \int_{(0,\infty)\setminus I_{h_*}} \Rev_{\mu_{h_*}}(v) d\lambda(v) \\
    &\le \lambda(I_{h_*})\Rev_{\mu_{h_*}}^* + (2^{2N-2}+C_N)B^{N-2},
\end{align*}
where the last inequality holds by using the trivial bound $\Rev_{\mu_{h_*}}(v)\le \Rev_{\mu_{h_*}}^*$ on $I_{h_*}$ and by \Cref{lem:N-class-revenue-localization} outside $I_{h_*}$. We conclude by using again \Cref{lem:N-class-revenue-localization} to lower bound $\Rev_{\mu_{h_*}}^*$,
\[
    \frac{\Rev_{\mu_{h_*}}(\lambda)}{\Rev_{\mu_{h_*}}^*} \le \lambda(I_{h_*}) + \frac{(2^{2N-2}+C_N)B^{N-2}}{B^{N-1}/4} \le \frac{1}{N}+\frac{4(2^{2N-2}+C_N)}{B} < \frac{\varepsilon}{2}+\frac{\varepsilon}{2}=\varepsilon.
\]
This proves the claim.
\end{proof}

\paragraph{Relation to earlier many-scale pricing lower bounds.}
Once \Cref{lem:N-class-revenue-localization} has localized the revenue of each \(\mu_h\), the final averaging step is very close to Claim A.1 of Babaioff et al.~\cite{BabaioffBDS17}: prices are divided into geometrically separated intervals, one selects an interval receiving little probability under the price lottery, and prices outside that interval are charged to a multiplicative revenue loss.  The difference lies in the hard instances. Babaioff et al.\ may use unrelated value distributions---in the relevant claim, deterministic values at the selected scales---whereas our distributions \(\mu_0,\ldots,\mu_{N-1}\) are fixed members of one common full moment class.  The role of \Cref{thm:nway-separator} and \Cref{lem:N-class-revenue-localization} is precisely to preserve all moments while retaining the off-interval revenue bound needed by that argument.
A continuous version of the same many-scale obstruction appears in \cite{GiannakopoulosPocasTsigonias2022}, who mix two-point laws with a common mean but different monopoly prices so that the resulting mixture is a truncated equal-revenue distribution.  In contrast, their higher moments vary across the hard family, whereas all members of our finite family agree on every moment.

The $N$-way construction yields the randomized-pricing lower bound by hiding the appropriate threshold region among many disjoint alternatives. In the next section, its first-atom concentration property lets us mix laws concentrated at different grid points to obtain hard i.i.d.\ instances for secretary success probability.

%% file: figures/pricing_localization.tex
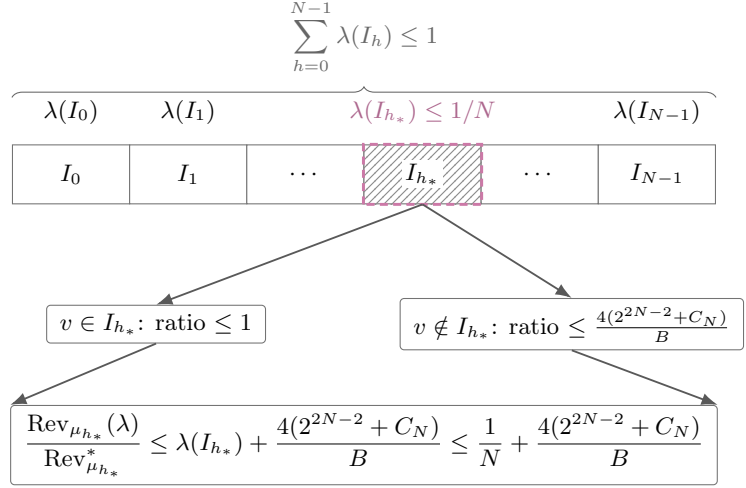
\begin{figure}[t]
\centering
\begin{tikzpicture}[
    paperfig,
    interval/.style={
        draw=black!55,
        fill=white,
        minimum width=1.55cm,
        minimum height=0.78cm,
        inner sep=2pt
    }
]
    \matrix (regions) [
        matrix of nodes,
        nodes={interval,anchor=center},
        column sep=-\pgflinewidth
    ] {
        $I_0$ & $I_1$ & $\cdots$ & $I_{h_*}$ & $\cdots$ & $I_{N-1}$ \\
    };

    \node[above=4pt of regions-1-1] {$\lambda(I_0)$};
    \node[above=4pt of regions-1-2] {$\lambda(I_1)$};
    \node[above=4pt of regions-1-4,text=figpurple!85!black]
        {$\lambda(I_{h_*})\le1/N$};
    \node[above=4pt of regions-1-6] {$\lambda(I_{N-1})$};

    \fill[pattern=north east lines,pattern color=black!45]
        (regions-1-4.south west) rectangle (regions-1-4.north east);
    \draw[figselected]
        (regions-1-4.south west) rectangle (regions-1-4.north east);
    \node[fill=white,inner sep=1pt] at (regions-1-4.center) {$I_{h_*}$};

    \draw[decorate,decoration={brace,amplitude=4pt},figgray]
        ($(regions-1-1.north west)+(0,0.62)$) --
        ($(regions-1-6.north east)+(0,0.62)$)
        node[midway,above=5pt]
        {$\displaystyle \sum_{h=0}^{N-1}\lambda(I_h)\le1$};

    \node[figbox] (inside)
        at ($(regions.south)+(-2.75,-1.45)$)
        {$v\in I_{h_*}$: ratio $\le1$};
    \node[figbox] (outside)
        at ($(regions.south)+(2.75,-1.45)$)
        {$v\notin I_{h_*}$: ratio $\le\frac{4(2^{2N-2}+C_N)}{B}$};
    \node[figbox] (bound)
        at ($(regions.south)+(0,-3.05)$)
        {$\displaystyle
          \frac{\Rev_{\mu_{h_*}}(\lambda)}
               {\Rev_{\mu_{h_*}}^*}
          \le \lambda(I_{h_*})+\frac{4(2^{2N-2}+C_N)}{B}
          \le \frac1N+\frac{4(2^{2N-2}+C_N)}{B}$};

    \draw[figarrow]
        (regions-1-4.south) -- (inside.north);
    \draw[figarrow]
        (regions-1-4.south) -- (outside.north);
    \draw[figarrow]
        (inside.south) -- (bound.north west);
    \draw[figarrow]
        (outside.south) -- (bound.north east);
\end{tikzpicture}
\caption{Randomized pricing: one disjoint region receives lottery mass at most $1/N$.}
\label{fig:pricing-localization}
\end{figure}

%% file: sections/secretary.tex
\section{Selection Ambiguity: The Secretary Problem}\label{sec:secretary}

We study the probability of selecting a maximum in the labeled cardinal-value model defined in \Cref{sec:preliminaries}. When the independent distributions are known in full, the worst-case finite-horizon success probability equals the full-information i.i.d.\ value \cite{EsfandiariEtAl2020,Nuti2022}; the latter converges to roughly $0.5801$ \cite{GilbertMosteller1966}. Replacing full distributional knowledge by exact moment sequences lowers the robust value to $s_m$, even when the values are i.i.d.\ This value converges to $1/e$ (\Cref{cor:secretary-asymptotic}).

The upper bound uses the following consequence of the $N$-way construction:
each residue-class distribution can be concentrated near a different
designated value while all distributions agree on all moments.

\begin{corollary}\label{cor:first-atom-concentration}
Fix an integer $N\ge2$, and for each integer $B\ge2$ let
$\mu_0,\dots,\mu_{N-1}$ be the distributions from
\Cref{thm:nway-separator}. Then
\[
\min_{0\le h\le N-1}\mu_h(\{B^h\})\longrightarrow 1
\qquad\text{as }B\to\infty.
\]
\end{corollary}

\begin{proof}

Fix $N$. For each $h\in\{0,\dots,N-1\}$, choosing the monomial $x^h$
from the $m=0$ factor and the constant term from every other factor
shows that $w_h\ge1$. It therefore suffices to show that the common
normalizer tends to one.

Use the zero residue class to write $W=\sum_{k\ge0}w_{Nk}$.
The exponent bounds in \Cref{thm:nway-separator} give $e(Nk)\ge k$:
this is immediate for $k=0$, and for $k\ge1$ it follows from
$Nk-e(Nk)\le N-1$. Thus, for $B\ge2^{N+1}$,
\[
1\le W=\sum_{k\ge0}w_{Nk}
\le\sum_{k\ge0}\left(\frac{2^N}{B}\right)^k
=\frac{1}{1-2^N/B}\longrightarrow1.
\]
Since $\mu_h$ is a probability measure, we have, uniformly in $h$,
\[
\frac1W\le\mu_h(\{B^h\})=\frac{w_h}{W}\le1.
\]
Taking the minimum over $h$ proves the claim.

\end{proof}

The next lemma finds a large subgrid on which a cardinal-value policy behaves almost like a rule based on relative ranks. Uniform i.i.d.\ sampling from that subgrid gives a hard instance; the additional error accounts for repeated values.

\begin{lemma}\label{lem:ramsey-erasure}
Fix integers $k\ge m\ge1$ and $\varepsilon>0$. There exists an integer $N_0=N_0(m,k,\varepsilon)\ge k$ such that, for every integer $N\ge N_0$, every ordered grid
\[
v_0<v_1<\cdots<v_{N-1},
\]
and every randomized cardinal-value secretary policy $\pi$ with horizon $m$, there is a set $H\subseteq\{0,\dots,N-1\}$ of size $k$ such that, when the $m$ labeled values are drawn independently and uniformly from $\{v_h:h\in H\}$ and arrive in uniformly random order,
\[
\P_\pi(\text{selects a maximum})\le s_m+\varepsilon+\frac{\binom m2}{k}.
\]
\end{lemma}

\begin{proof}
Represent $\pi$ by its conditional stopping probability after each history of observed labels and values. Choose a mesh size $\delta>0$ with $m\delta\le\varepsilon$, and round each stopping probability down to a multiple of $\delta$. For each $t=1,\dots,m$, color every $t$-element subset of the grid by the finite vector of rounded stopping probabilities over all orderings of those values and all ordered $t$-tuples of distinct labels.

Successive applications of the finite Ramsey theorem give a $k$-element set $H$ that is homogeneous for every history length, provided $N$ is sufficiently large as a function of $m,k,\delta$. On histories of distinct grid values indexed by $H$, the rounded rule $\widetilde\pi$ depends only on time, the ordered label history, and the relative order of the observed values. Couple $\pi$ and $\widetilde\pi$ using the same uniform random choices. For any fixed sequence of distinct values and labels, their decisions differ with probability at most $m\delta$.

Let $\mathcal D$ be the event that the $m$ independent uniform draws are distinct. Conditional on any unordered set of $m$ distinct draws, their assignment to labels is uniform and their arrival order is independent and uniform. The observed label order is therefore independent of the value ranks. On this event, $\widetilde\pi$ is a randomized classical secretary rule, with labels supplying only independent auxiliary randomness. Its success probability is at most $s_m$ \cite{GilbertMosteller1966,Ferguson1989}, so $\pi$ succeeds with conditional probability at most $s_m+m\delta$.

Each pair of draws agrees with probability $1/k$, and hence $\P(\mathcal D^c)\le\binom m2/k$. It follows that
\[
\P_\pi(\text{success})\le s_m+m\delta+\P(\mathcal D^c)\le s_m+\varepsilon+\frac{\binom m2}{k}.
\]
\end{proof}

\begin{theorem}\label{thm:full-secretary-success}
For every integer horizon $m\ge1$, the optimal robust success probability under exact moment disclosure in the labeled cardinal-value random-order secretary problem is $s_m$, both for independent values and for i.i.d.\ values.
\end{theorem}

\begin{proof}
For the lower bound, use a classical optimal skip-then-record rule, breaking value ties by independent continuous tie-breakers. Conditional on any realized values, this rule selects the lexicographic maximum with probability $s_m$; that item is one of the value maximizers. Hence the robust success probability is at least $s_m$ in both models.

For the upper bound, we construct an i.i.d.\ instance by mixing moment-equivalent laws concentrated at the grid points of \Cref{lem:ramsey-erasure}. Fix $\varepsilon>0$ and choose an integer $k\ge m$ such that $\binom m2/k\le\varepsilon/3$. Take an integer $N\ge\max\{2,N_0(m,k,\varepsilon/3)\}$. By \Cref{cor:first-atom-concentration}, choose an integer $B\ge2$ large enough that the $N$-way construction satisfies
\[
\mu_h(\{B^h\})\ge 1-\frac{\varepsilon}{3m}
\qquad\text{for every }h=0,\dots,N-1.
\]
All uniform mixtures of $k$ members of $\{\mu_0,\dots,\mu_{N-1}\}$ share the same full moment sequence. This disclosure is fixed before the policy is chosen, since $N$ and $B$ depend only on $m$ and $\varepsilon$. Fix a moment-based policy $\pi$ under this disclosure, allowing it to use the fact that the values are i.i.d.

Apply \Cref{lem:ramsey-erasure} to $\pi$ on the ordered grid $B^0<B^1<\cdots<B^{N-1}$. It gives a set $H$ of size $k$ such that $\pi$ succeeds with probability at most $s_m+2\varepsilon/3$ on independent uniform draws from $\{B^h:h\in H\}$. Define
\[
\psi_H:=\frac1k\sum_{h\in H}\mu_h.
\]
This distribution is compatible with the disclosure and has finite moments of every order.

To compare the two instances, draw $h_1,\dots,h_m$ independently and uniformly from $H$. Set $Y_i=B^{h_i}$ and, conditional on these indices, draw $X_i\sim\mu_{h_i}$ independently. Then the $X_i$ are i.i.d.\ with law $\psi_H$, and
\[
\P(\text{some }X_i\ne Y_i)\le\sum_{i=1}^m\P(X_i\ne B^{h_i})\le\frac{\varepsilon}{3}.
\]
Run $\pi$ on both vectors with the same disclosure, arrival order, and private randomization. Whenever the vectors agree, the histories and the success events agree. Therefore
\[
\P_\pi(\text{success on }X)\le\P_\pi(\text{success on }Y)+\frac{\varepsilon}{3}\le s_m+\varepsilon.
\]
Since $\pi$ and $\varepsilon$ were arbitrary, the robust value in the i.i.d.\ model is at most $s_m$. The same upper bound holds in the independent model, whose admissible instances include the i.i.d.\ ones. Together with the lower bound, this proves the theorem.
\end{proof}

\begin{corollary}\label{cor:secretary-asymptotic}
In both the independent and the i.i.d.\ models, the finite-horizon robust success values satisfy
\[
\lim_{m\to\infty}s_m=\frac1e.
\]
\end{corollary}

\begin{proof}
By \Cref{thm:full-secretary-success}, the robust value at every finite horizon $m\ge1$ is $s_m$ in both models. The classical limit recalled in \Cref{sec:defn-problems} gives the displayed convergence.
\end{proof}

The next section treats prophet inequalities, where the binary construction separates the probabilities of seeing enough large observations to change an order-statistic benchmark.

%% file: sections/prophet.tex
\section{Order-Statistic Ambiguity: Prophet Inequalities}\label{sec:prophet}

The hard instances of Correa et al.~\cite{CorreaCristiLivanosVerdugoZhang26} are stated for the expected-maximum benchmark, but their construction can also be adapted to every fixed order statistic. Each active block contains, with high probability, many copies of its designated value. The fixed-$r$ logarithmic barrier can therefore also be obtained by adapting their construction.

The binary separator yields the required order-statistic gap directly from its designated-atom and tail estimates. The resulting proof is modular and uses the same construction that drives the quantile and newsvendor lower bounds.

For integers $m\ge r\ge1$, let $\rho^{\mathrm{mom}}_{m,r}$ denote the
supremum of all $\alpha\ge0$ for which there is a randomized moment-based
single-choice policy such that, for every collection $Z_1,\dots,Z_m$ of
independent nonnegative random variables with finite moments of every order,
the policy's stopping time $\tau$ satisfies
\[
\E[Z_\tau]\ge \alpha\,\E[Z_{(r)}]
\]
under the disclosed moment sequences. This guarantee-based definition also covers
instances with zero benchmark, without forming an indeterminate ratio. In this
section, $\log$ denotes the natural logarithm.

The following consequence of the binary separator provides exactly the
order-statistic gap needed for the upper bound.

\begin{corollary}\label{cor:order-stat-separator}

For every fixed integer $r\ge1$ there is a constant $\kappa_r>0$ such that the following holds. Let $B,n$ be integers with $B,n\ge3$ and $n\ge\max\{2B,2r\}$, and let $\mu,\nu$ be the distributions from \Cref{thm:binary-separator}. If $X_1,\dots,X_n$ are i.i.d.\ from $\nu$ and $Y_1,\dots,Y_n$ are i.i.d.\ from $\mu$, then
\[
\E_\nu[X_{(r)}]\ge \kappa_r nB,
\qquad
\E_\mu[Y_{(1)}]\le 2B.
\]
In addition,
\[
\E_\nu\!\left[\sum_{i=1}^n X_i\mathbf 1_{\{X_i\ge x_3\}}\right]\le\frac{6}{n},
\qquad
\E_\mu\!\left[\sum_{i=1}^n Y_i\mathbf 1_{\{Y_i\ge x_3\}}\right]\le\frac{6}{n}.
\]

\end{corollary}

\begin{proof}

Let $H:=\sum_{i=1}^n\mathbf 1_{\{X_i=x_2\}}$, and let
$H_0\sim\operatorname{Bin}(n,1/(4n))$. Since $p_2\ge1/(4n)$, the variable
$H$ stochastically dominates $H_0$, and hence
\[
\P(H\ge r)\ge\P(H_0=r)=
\binom{n}{r}\frac{1}{(4n)^r}\left(1-\frac{1}{4n}\right)^{n-r}
\ge \frac{e^{-1/4}}{8^r r!}=: \kappa_r,
\]
using $n\ge2r$. Here
$\binom nr\ge(n/2)^r/r!$, while
$\log(1-x)\ge-x/(1-x)$ with $x=1/(4n)$ gives
$(1-1/(4n))^{n-r}\ge e^{-1/4}$. On the event $H\ge r$,
$X_{(r)}\ge x_2=nB$.

For the tails, \Cref{thm:binary-separator} gives
\[
n\sum_{j\ge3} x_j\frac{|w_j|}{W}
\le
8(nB)^2\sum_{j\ge3}\frac1{x_j}
\le
\frac{16(nB)^2}{x_3}
\le \frac{16}{nB}
<\frac{6}{n},
\]
because $x_3=(nB)^B$ and $B\ge3$. This proves both displayed tail bounds. Finally, every non-tail sample from $\mu$ equals $B$, so
\[
\E_\mu[Y_{(1)}]\le B+
\E_\mu\!\left[\sum_{i=1}^n Y_i\mathbf 1_{\{Y_i\ge x_3\}}\right]\le 2B
\]
for the stated parameter range.

\end{proof}

\begin{theorem}\label{thm:prophet}

For every integer $r\ge1$, there exist constants $a_r,b_r>0$ and an integer
$m_0(r)\ge\max\{r,2\}$ such that, for every integer $m\ge m_0(r)$,
\[
\frac{a_r}{\log m}
\le \rho^{\mathrm{mom}}_{m,r}
\le \frac{b_r}{\log m}.
\]

\end{theorem}

For the lower-bound direction we use the following theorem of Correa et
al.~\cite{CorreaCristiLivanosVerdugoZhang26}, stated directly in terms of the
guarantee above.

\begin{proposition}\label{prop:log-lowerbound}

There are a constant $a>0$ and an integer $m_{\mathrm{lb}}\ge2$ such that,
for every integer $m\ge m_{\mathrm{lb}}$,
\[
\rho^{\mathrm{mom}}_{m,1}\ge \frac{a}{\log m}.
\]

\end{proposition}

\begin{proof}[Proof of \Cref{thm:prophet}]

The lower bound follows from \Cref{prop:log-lowerbound}: since
$Z_{(1)}\ge Z_{(r)}$ pointwise, a guarantee against the maximum is also
one against the $r$-th largest value. We may therefore take $a_r:=a$.

For the upper bound, we arrange the binary separator pairs in blocks of
doubling size. At base $B=3$, a block of $n$ draws from $\nu_n$ has
an order-statistic benchmark proportional to $n$, whereas the values
below the tail threshold $x_3$ in its moment-equivalent $\mu_n$ block
are all $3$. The hard
instances switch from $\nu$-blocks to $\mu$-blocks at different points.
A policy must allocate its stopping probability across the growing
blocks without learning the switch point from the disclosed moments.

Let $\kappa_r$ be the constant in \Cref{cor:order-stat-separator} and take
$N_r:=\lceil\max\{6,2r,1/\kappa_r\}\rceil$, so that the separator
estimates apply to every block size $n\ge N_r$. Write $(\mu_n,\nu_n)$
for the corresponding pair with base $3$. Fix $m\ge4N_r^2$ and let $L$
be the largest integer with $N_r(2^L-1)\le m$. We can then fit $L$
blocks of sizes $n_j=N_r2^{j-1}$ into the horizon. In instance
$\mathcal I_t$, blocks $1,\dots,t$ contain independent draws from
$\nu_{n_j}$, blocks $t+1,\dots,L$ contain independent draws from
$\mu_{n_j}$, and any remaining coordinates are zero. All coordinates
are independent. Every instance discloses the same moment sequence
at each coordinate, as illustrated in \Cref{fig:prophet-hard-instances}.

\input{figures/prophet_hard_instances}

Fix any randomized moment-based stopping rule $\tau$. On the all-$\nu$
instance $\mathcal J:=\mathcal I_L$, let
$a_j:=\P_{\mathcal J}(\tau\text{ stops in block }j)$.
These probabilities satisfy $\sum_{j=1}^L a_j\le1$. Through any block
$j\le t$, the history has the same law under $\mathcal I_t$ as under
$\mathcal J$, so the stopping probability in that block is still $a_j$.
Ordinary rewards in a $\nu$-block are at most $3n_j$; after the switch,
they are at most $3$. Adding the tail bounds from
\Cref{cor:order-stat-separator} gives
\[
\E_{\mathcal I_t}[Z_\tau]
\le 3\sum_{j\le t}n_ja_j+3+6\sum_{j=1}^L\frac1{n_j}
\le 3\sum_{j\le t}n_ja_j+3+\frac{12}{N_r}.
\]
Here the last inequality uses the geometric sum
$\sum_{j=1}^L2^{1-j}<2$.

To compare these rewards across scales, average the switch points with
weights $\pi_t=2^{-t}/Z_L$, where
$Z_L:=\sum_{t=1}^L2^{-t}\in[1/2,1)$. The decreasing weights offset the
growing block sizes: $\sum_{t=j}^L\pi_t\le2^{2-j}$, while
$n_j=N_r2^{j-1}$. The total stopping probability is at most one, so the
weighted policy reward stays bounded. In contrast, block $t$ alone
gives $\E_{\mathcal I_t}[Z_{(r)}]\ge3\kappa_r n_t$, and each switch
point contributes the same amount to the weighted benchmark. Thus
\begin{align*}
\sum_{t=1}^L\pi_t\,\E_{\mathcal I_t}[Z_\tau]
&\le 6N_r+3+\frac{12}{N_r}=:A_r,\\
\sum_{t=1}^L\pi_t\,\E_{\mathcal I_t}[Z_{(r)}]
&\ge \frac{3\kappa_rN_rL}{2Z_L}
\ge \frac{3\kappa_rN_rL}{2}.
\end{align*}
Consequently, at least one switch point $t$ satisfies
\[
\E_{\mathcal I_t}[Z_\tau]
\le \frac{2A_r}{3\kappa_rN_rL}\,
\E_{\mathcal I_t}[Z_{(r)}].
\]
It remains to relate the number of blocks to the horizon. Maximality of
$L$ gives $N_r(2^{L+1}-1)>m$, and $m\ge4N_r^2$ then implies
\[
L>\log_2(m/N_r)-1\ge\tfrac12\log_2m
=\frac{\log m}{2\log2}.
\]
The reward on the selected instance is therefore at most
$b_r/\log m$ times its benchmark, with the explicit constants
\[
b_r:=\frac{4A_r\log2}{3\kappa_rN_r},
\qquad
m_0(r):=\max\{r,m_{\mathrm{lb}},4N_r^2\}.
\]
This proves the upper bound for every $m\ge m_0(r)$ and completes the theorem.

\end{proof}

For the multi-selection variant, we take the reward to be the total value of the selected items. This is the stronger reward model for an upper bound; the same conclusion then also applies when the reward is only the maximum selected value.

\begin{corollary}\label{cor:prophet-r-selection}

Fix an integer $r\ge1$, and let $b_r$ and $m_0(r)$ be upper-bound constants from
\Cref{thm:prophet}. For every integer $m\ge m_0(r)$ and every moment-based algorithm
that selects at most $r$ values and receives their total value, there is a
length-$m$ instance $Z_1,\dots,Z_m$ of independent nonnegative random variables with finite
moments of every order on which the algorithm's expected reward is at most
\[
\frac{r b_r}{\log m}\,\E[Z_{(r)}].
\]

\end{corollary}

\begin{proof}

Given an algorithm that selects at most $r$ values, sample $J$ uniformly from $\{1,\dots,r\}$ before the arrivals and stop when the algorithm makes its $J$-th selection, receiving value $0$ if it makes fewer than $J$ selections. Conditional on the algorithm's selected values, the resulting single-choice algorithm obtains exactly $1/r$ of their total value in expectation over $J$. Apply the explicit upper-bound construction in the proof of \Cref{thm:prophet} to this induced policy. Its witness instance, after multiplying the inequality by $r$, gives the claim.

\end{proof}

This completes the lower-bound applications of the construction toolkit.

%% file: figures/prophet_hard_instances.tex
\begin{figure}[H]
\centering
\begin{tikzpicture}[
    font=\footnotesize,
    nu/.style={draw=figblue!85,fill=figblue!12,line width=0.65pt},
    mu/.style={draw=figorange!90,fill=figorange!12,dashed,line width=0.75pt},
    zero/.style={
        draw=figgray!75,
        fill=figgray!8,
        pattern=north east lines,
        pattern color=figgray!45
    }
]
    \def\rowh{0.66}
    \def\yone{3.15}
    \def\ytwo{2.22}
    \def\yt{1.29}
    \def\yL{0.36}

    \node[font=\scriptsize,text=figgray] at (-0.38,4.08) {block};
    \foreach \x/\lab in {
        0.65/{1},
        1.95/{2},
        3.05/{\cdots},
        4.25/{t},
        5.75/{t+1},
        7.00/{\cdots},
        8.30/{L}
    }{
        \node[align=center] at (\x,4.08) {$\lab$};
    }
    \node[align=center] at (9.72,4.08) {zero\\padding};

    \node[left,font=\small] at (-0.12,\yone) {$\mathcal I_1$};
    \node[left,font=\small] at (-0.12,\ytwo) {$\mathcal I_2$};
    \node[left,font=\small] at (-0.12,\yt) {$\mathcal I_t$};
    \node[left,font=\small] at (-0.12,\yL) {$\mathcal J=\mathcal I_L$};

    \draw[nu] (0,\yone-\rowh/2) rectangle (1.3,\yone+\rowh/2)
        node[midway] {$\nu$};
    \draw[mu] (1.3,\yone-\rowh/2) rectangle (2.6,\yone+\rowh/2)
        node[midway] {$\mu$};
    \draw[mu] (2.6,\yone-\rowh/2) rectangle (3.5,\yone+\rowh/2)
        node[midway] {$\cdots$};
    \draw[mu] (3.5,\yone-\rowh/2) rectangle (5.0,\yone+\rowh/2)
        node[midway] {$\mu$};
    \draw[mu] (5.0,\yone-\rowh/2) rectangle (6.5,\yone+\rowh/2)
        node[midway] {$\mu$};
    \draw[mu] (6.5,\yone-\rowh/2) rectangle (7.5,\yone+\rowh/2)
        node[midway] {$\cdots$};
    \draw[mu] (7.5,\yone-\rowh/2) rectangle (9.1,\yone+\rowh/2)
        node[midway] {$\mu$};
    \draw[zero] (9.1,\yone-\rowh/2) rectangle (10.35,\yone+\rowh/2)
        node[midway,fill=white,inner sep=1pt] {$0$};

    \draw[nu] (0,\ytwo-\rowh/2) rectangle (1.3,\ytwo+\rowh/2)
        node[midway] {$\nu$};
    \draw[nu] (1.3,\ytwo-\rowh/2) rectangle (2.6,\ytwo+\rowh/2)
        node[midway] {$\nu$};
    \draw[mu] (2.6,\ytwo-\rowh/2) rectangle (3.5,\ytwo+\rowh/2)
        node[midway] {$\cdots$};
    \draw[mu] (3.5,\ytwo-\rowh/2) rectangle (5.0,\ytwo+\rowh/2)
        node[midway] {$\mu$};
    \draw[mu] (5.0,\ytwo-\rowh/2) rectangle (6.5,\ytwo+\rowh/2)
        node[midway] {$\mu$};
    \draw[mu] (6.5,\ytwo-\rowh/2) rectangle (7.5,\ytwo+\rowh/2)
        node[midway] {$\cdots$};
    \draw[mu] (7.5,\ytwo-\rowh/2) rectangle (9.1,\ytwo+\rowh/2)
        node[midway] {$\mu$};
    \draw[zero] (9.1,\ytwo-\rowh/2) rectangle (10.35,\ytwo+\rowh/2)
        node[midway,fill=white,inner sep=1pt] {$0$};

    \draw[nu] (0,\yt-\rowh/2) rectangle (1.3,\yt+\rowh/2)
        node[midway] {$\nu$};
    \draw[nu] (1.3,\yt-\rowh/2) rectangle (2.6,\yt+\rowh/2)
        node[midway] {$\nu$};
    \draw[nu] (2.6,\yt-\rowh/2) rectangle (3.5,\yt+\rowh/2)
        node[midway] {$\cdots$};
    \draw[nu] (3.5,\yt-\rowh/2) rectangle (5.0,\yt+\rowh/2)
        node[midway] {$\nu$};
    \draw[figselected]
        (3.5,\yt-\rowh/2) rectangle (5.0,\yt+\rowh/2);
    \draw[mu] (5.0,\yt-\rowh/2) rectangle (6.5,\yt+\rowh/2)
        node[midway] {$\mu$};
    \draw[mu] (6.5,\yt-\rowh/2) rectangle (7.5,\yt+\rowh/2)
        node[midway] {$\cdots$};
    \draw[mu] (7.5,\yt-\rowh/2) rectangle (9.1,\yt+\rowh/2)
        node[midway] {$\mu$};
    \draw[zero] (9.1,\yt-\rowh/2) rectangle (10.35,\yt+\rowh/2)
        node[midway,fill=white,inner sep=1pt] {$0$};

    \draw[nu] (0,\yL-\rowh/2) rectangle (1.3,\yL+\rowh/2)
        node[midway] {$\nu$};
    \draw[nu] (1.3,\yL-\rowh/2) rectangle (2.6,\yL+\rowh/2)
        node[midway] {$\nu$};
    \draw[nu] (2.6,\yL-\rowh/2) rectangle (3.5,\yL+\rowh/2)
        node[midway] {$\cdots$};
    \draw[nu] (3.5,\yL-\rowh/2) rectangle (5.0,\yL+\rowh/2)
        node[midway] {$\nu$};
    \draw[nu] (5.0,\yL-\rowh/2) rectangle (6.5,\yL+\rowh/2)
        node[midway] {$\nu$};
    \draw[nu] (6.5,\yL-\rowh/2) rectangle (7.5,\yL+\rowh/2)
        node[midway] {$\cdots$};
    \draw[nu] (7.5,\yL-\rowh/2) rectangle (9.1,\yL+\rowh/2)
        node[midway] {$\nu$};
    \draw[zero] (9.1,\yL-\rowh/2) rectangle (10.35,\yL+\rowh/2)
        node[midway,fill=white,inner sep=1pt] {$0$};

    \node[draw=figgray!60,fill=white,rounded corners=2pt,
          align=center,text width=4.1cm,inner xsep=4pt,inner ysep=3pt]
        at (2.15,-0.72)
        {block $j$: $n_j=N_r2^{j-1}$ coordinates\\(not to scale)};
    \node[figselected,fill=white,rounded corners=2pt,
          align=center,text width=4.35cm,inner xsep=4pt,inner ysep=3pt]
        at (7.9,-0.72)
        {benchmark block $t$\\[-1pt]
         $\displaystyle \E_{\mathcal I_t}[Z_{(r)}]\ge3\kappa_rn_t$};
    \node[draw=figgray!60,fill=white,rounded corners=2pt,
          align=center,inner xsep=5pt,inner ysep=3pt]
        at (5.17,-1.72)
        {$\mu_{n_j}$ and $\nu_{n_j}$ agree on all moments in every block $j$};
\end{tikzpicture}
\caption{Prophet hard instances: the switch point varies while block sizes double.}
\label{fig:prophet-hard-instances}
\end{figure}
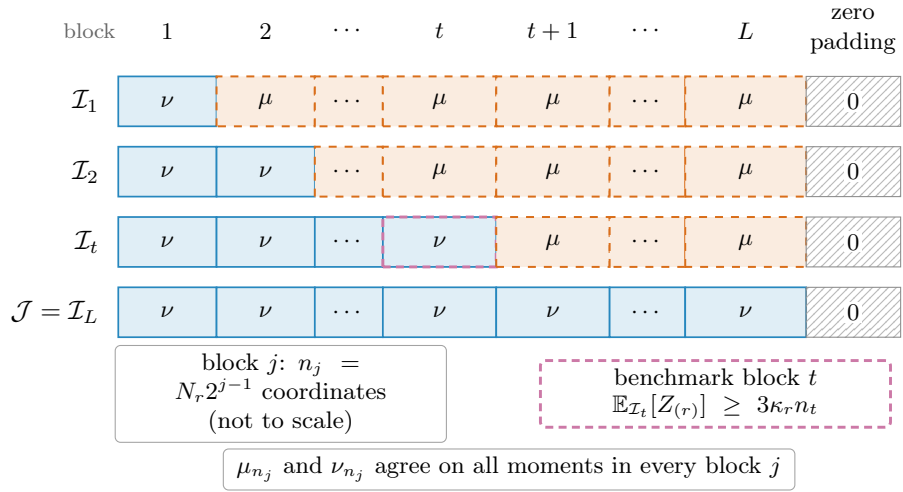

%% file: sections/proofs.tex
\section{Proofs of the Constructions}
\label{sec:proofs}

This section supplies the deferred proofs of the two separators in
\Cref{sec:impossibility}. We first establish the limiting construction and
decay estimates for the binary separator, followed by the analytic
estimates for the $N$-way residue-class separator.

\subsection{Proof of the Binary Separator}
The goal of this section is to prove \Cref{prop:key-prop-1}. We restate an extended version first.

\begin{proposition}
\label{prop:key-prop-1_extended}
Fix integers $n,B\ge 3$. Let
\[
x_0=1,\qquad x_1=B,\qquad x_j=(nB)^{B^{j-2}}\quad (j\ge 2).
\]
For each $K\ge 0$, let $(w_j^{(n,K)})_{j=0}^{K+2}$ be the unique solution of
\[
\sum_{j=0}^{K+2} w_j^{(n,K)}x_j^s=0 \quad (s=0,\dots,K+1),
\qquad
w_2^{(n,K)}=1.
\]
For each fixed $j\ge 0$, the quantity $w_j^{(n,K)}$ is defined for all
$K\ge \max\{0,j-2\}$. Define
\[
w_j:=\lim_{K\to\infty} w_j^{(n,K)},
\]
where the limit is taken over integers $K\ge \max\{0,j-2\}$.

Set
\[
W^+ := \sum_{\substack{j\ge 0\\ j\text{ even}}} w_j,
\qquad
W^- := \sum_{\substack{j\ge 0\\ j\text{ odd}}} |w_j|.
\]
Then \(W^+=W^-=:W\in(0,\infty)\). Define probability measures \(\nu\) and \(\mu\) on
\(\{x_j:j\ge 0\}\) by
\[
\nu(\{x_j\}) :=
\begin{cases}
\dfrac{w_j}{W}, & j\text{ even},\\
0, & j\text{ odd},
\end{cases}
\qquad
\mu(\{x_j\}) :=
\begin{cases}
\dfrac{|w_j|}{W}, & j\text{ odd},\\
0, & j\text{ even}.
\end{cases}
\]
Then \(\mu\) and \(\nu\) agree on all moments, and all of these moments are finite.
Moreover, these distributions satisfy:
\begin{itemize}
    \item[\(\mathrm{(i)}\)] \(p_0 := \P_{\nu}(X=1) \ge 1 - \frac{2}{n}\) and
    \(q_0 := \P_{\mu}(Y=1) = 0\);
    \item[\(\mathrm{(ii)}\)] \(p_1 := \P_{\nu}(X=B) =0\) and
    \(q_1 := \P_{\mu}(Y=B) \ge 1 - \frac{1}{n^2B^2}\);
    \item[\(\mathrm{(iii)}\)] \(p_2:= \P_{\nu}(X=nB) \ge \frac{1}{4n}\) and
    \(q_2 := \P_{\mu}(Y=nB) = 0\);
    \item[\(\mathrm{(iv)}\)] \(W \ge n\) and
    $|w_j| \le \frac{8(nB)^2}{x_j^2}$ for all $j\in \N$.
\end{itemize}
\end{proposition}

In accordance with the statement, let us fix integers $B,n\ge 3$, and define
\[
x_0=1,\qquad x_1=B,\qquad x_j=(nB)^{B^{j-2}}\quad \forall j\ge 2.
\]
Note that $x_{j+1}=x_j^B\ge x_j^3$ for all $j\ge 3$. In the remainder, we first construct finite-$K$ barycentric weights $w_j^{(n,K)}$, then pass to the limit $K\to\infty$. After proving existence of the limit in \Cref{lem:limit-weights-exist}, we will abbreviate $w_j:=w_j^{(n,\infty)}$ for $j \in \N$, which is the notation used in the statement of \Cref{prop:key-prop-1}. 

\subsubsection{Constructing Distributions Agreeing on Finitely Many Moments} 
\label{subsubsec:finite-moment-construction}
We now construct, for each integer $K\ge 0$, a pair of distributions $\mu, \nu$ matching (finite) moments through order $K+1$. To this end, consider the linear system
\[
    \sum_{j=0}^{K+2} w_j^{(n,K)} x_j^s=0,
    \qquad s\in\{0,1,\dots,K+1\},
\]
together with the normalization $w_2^{(n,K)}=1$.
Since the nodes $x_0,\dots,x_{K+2}$ are distinct, the matrix
\[
V_K=(x_j^s)_{\substack{0\le s\le K+1\\0\le j\le K+2}}
\in \mathbb R^{(K+2)\times (K+3)}
\]
has rank $K+2$. Hence $\ker(V_K)$ is one-dimensional. The vector
\[
\lambda^{(n,K)}:=\bigl(\lambda_j^{(n,K)}\bigr)_{j=0}^{K+2},
\qquad
\lambda_j^{(n,K)}=
\left(\prod_{\substack{0\le i\le K+2\\ i\neq j}}(x_j-x_i)\right)^{-1},
\]
lies in $\ker(V_K)$ and satisfies $\lambda_2^{(n,K)}\neq 0$ by the Lagrange interpolation identity below.

Therefore there is a unique null vector normalized by $w_2^{(n,K)}=1$, namely $w_j^{(n,K)}=\frac{\lambda_j^{(n,K)}}{\lambda_2^{(n,K)}}$.
With this notation,
\[
w_j^{(n,K)}
=
\frac{\lambda_j^{(n,K)}}{\lambda_2^{(n,K)}}
=
\frac{\displaystyle\prod_{i\neq 2}(x_2-x_i)}
{\displaystyle\prod_{i\neq j}(x_j-x_i)}.
\]
Indeed, for every polynomial $p$ of degree at most $K+1$, the Lagrange interpolation formula
\[
p(t)=\sum_{j=0}^{K+2} p(x_j)\prod_{\substack{0\le i\le K+2\\ i\neq j}}
\frac{t-x_i}{x_j-x_i}
\]
shows, by comparing the coefficient of $t^{K+2}$, that $\sum_{j=0}^{K+2}\lambda_j^{(n,K)}p(x_j)=0$.
Applying this with $p(t)=t^s$, $s=0,\dots,K+1$, gives the annihilation identities, and normalizing by $\lambda_2^{(n,K)}$ enforces $w_2^{(n,K)}=1$.

The sign pattern is explicit. The numerator $\prod_{i\neq 2}(x_2-x_i)$ has exactly $K$ negative factors, while the denominator $\prod_{i\neq j}(x_j-x_i)$ has exactly $K+2-j$ negative factors. Hence $\operatorname{sgn}(w_j^{(n,K)}) = (-1)^{K-(K+2-j)} = (-1)^{j-2}=(-1)^{j}$.
Accordingly, define
\[
I_+:=\{0 \le j\le K+2:\ j \text{ even}\},
\qquad
I_-:=\{0 \le j\le K+2:\ j \text{ odd}\}.
\]
Then, for every $K\ge 0$,
\[
\operatorname{sgn}(w_j^{(n,K)})>0 \iff j\in I_+,
\qquad
\operatorname{sgn}(w_j^{(n,K)})<0 \iff j\in I_-,
\qquad 0\le j\le K+2.
\]

The relation for $s=0$ gives $\sum_{j=0}^{K+2} w_j^{(n,K)}=0$, so the positive and negative parts agree:
\[
W^{(n,K)}
:=
\sum_{\substack{0\le j\le K+2\\ j\in I_+}} w_j^{(n,K)}
=
\sum_{\substack{0\le j\le K+2\\ j\in I_-}} |w_j^{(n,K)}|
>0.
\]
We define two probability measures on $\{x_0,\dots,x_{K+2}\}$ by
\[
\nu^{(n,K)}(\{x_j\})
:=
\frac{w_j^{(n,K)}}{W^{(n,K)}} \mathbf 1_{\{j\in I_+\}},
\qquad
\mu^{(n,K)}(\{x_j\})
:=
\frac{|w_j^{(n,K)}|}{W^{(n,K)}} \mathbf 1_{\{j\in I_-\}}.
\]
These measures have disjoint support, and for every $s\in \{0,1,\dots,K+1\}$, $\sum_j x_j^s  \nu^{(n,K)}(\{x_j\}) = \sum_j x_j^s  \mu^{(n,K)}(\{x_j\})$.
Thus $\mu^{(n,K)}$ and $\nu^{(n,K)}$ match moments through order $K+1$.

\subsubsection{Extending to Infinity}

We now pass to the limit \(K\to\infty\) and show that the resulting limiting measures
\(\mu^{(n,\infty)}\) and \(\nu^{(n,\infty)}\) match moments of every order.

\begin{lemma}
\label{lem:limit-weights-exist}
For each fixed $j\ge 0$, the limit $w_j^{(n,\infty)}:=\lim_{K\to\infty} w_j^{(n,K)}$ exists.
\end{lemma}

\begin{proof}
Fix $j\ge 0$, and set $K_0:=\max\{0,j-2\}$. For $K\ge K_0$, when the new node $x_{K+3}$ is added, the normalized barycentric weights satisfy
\begin{equation}
\label{eq:update-impossibility}
w_j^{(n,K+1)}
=
w_j^{(n,K)}
\frac{x_2-x_{K+3}}{x_j-x_{K+3}}.
\end{equation}
Indeed, if
\[
\widetilde w_j^{(n,K)}
:=
\left(\prod_{\substack{0\le i\le K+2\\ i\neq j}}(x_j-x_i)\right)^{-1},
\]
then
\[
\widetilde w_j^{(n,K+1)}
=
\frac{\widetilde w_j^{(n,K)}}{x_j-x_{K+3}},
\qquad
\widetilde w_2^{(n,K+1)}
=
\frac{\widetilde w_2^{(n,K)}}{x_2-x_{K+3}},
\]
and dividing gives \eqref{eq:update-impossibility}.

Iterating \eqref{eq:update-impossibility} from $K_0$ to $K$, we obtain
\[
w_j^{(n,K)}
=
w_j^{(n,K_0)}
\prod_{\ell=K_0+3}^{K+2}
\frac{x_2-x_\ell}{x_j-x_\ell}.
\]
Since $\ell\ge K_0+3\ge \max\{j+1,3\}$, we have $x_\ell>x_j$, and therefore $\frac{x_2-x_\ell}{x_j-x_\ell} = \frac{x_\ell-x_2}{x_\ell-x_j} = 1+\frac{x_j-x_2}{x_\ell-x_j}$.

We claim that $x_\ell\ge 2\max\{x_j,x_2\}$ for all $\ell\ge K_0+3$. Indeed, if $j\le 2$, then $\ell\ge 3$, so $x_\ell\ge x_3=(nB)^B\ge (nB)^3\ge 2nB=2x_2\ge 2x_j$.
If $j\ge 3$, then $\ell\ge j+1$, so $x_\ell\ge x_{j+1}=x_j^B\ge x_j^3\ge 2x_j$,
and also $x_j\ge x_3\ge 2x_2$, hence $x_\ell\ge 2x_2$. Thus $x_\ell-x_j\ge \frac{x_\ell}{2}$,
and therefore $\left|\frac{x_j-x_2}{x_\ell-x_j}\right| \le \frac{2(x_j+x_2)}{x_\ell}$.

Because $x_{\ell+1}=x_\ell^B\ge x_\ell^3$, the series $\sum_{\ell\ge 3} x_\ell^{-1}$ converges. Hence $\sum_{\ell=K_0+3}^\infty \left| \frac{x_j-x_2}{x_\ell-x_j} \right| <\infty$.
Therefore the infinite product
\[
\prod_{\ell=K_0+3}^{\infty}
\left(
1+\frac{x_j-x_2}{x_\ell-x_j}
\right)
\]
converges to a finite, nonzero limit. Hence $w_j^{(n,K)}$ converges as $K\to\infty$, and $\operatorname{sgn}(w_j^{(n,\infty)}) = \operatorname{sgn}(w_j^{(n,K_0)}) = (-1)^j$, as desired.
\end{proof}

From now on, we write $w_j:=w_j^{(n,\infty)}$ for $j \ge 0$ in agreement with the notation of \Cref{prop:key-prop-1}. 

\subsubsection{Estimating Decay Rates for Limiting Weights}

We next derive quantitative decay estimates for the limiting weights. The normalization into limiting probability measures will be carried out only after absolute summability has been established. We first record a convenient tail estimate for the barycentric weights, to then prove that the distributions defined above match in all (infinite) moments:

\begin{lemma}
\label{lem:all-moment-weight-decay}
For every integer $M\ge 1$, define
\[
L_M:=\max\left\{0,\left\lceil \frac{M-2}{1-1/B}\right\rceil\right\}.
\]
Then, with
\[
    C_M^{(0,1,2)}
    :=
    \max\{ 2n, 2nB^M, (nB)^M \}, \qquad
    C_M^{(\ge 3)}
    :=
    \begin{cases}
        8(nB)^2 & M=1,2,\\[1ex]
        8(nB)^2\max\{ x_{L_M+2}^{ M-2}, 2^{L_M} \} & M\ge 3,
    \end{cases}
\]
the constant $C_M:=\max\{C_M^{(0,1,2)}, C_M^{(\ge 3)}\}$ satisfies
\[
    |w_j|
    \le
    \frac{C_M}{x_j^M}
    \qquad \forall j\ge 0.
\]
\end{lemma}

\begin{proof}
Fix $M\ge 1$ and write $L:=L_M$. We begin with the cases $j=0,1,2$. For $j=0$, using the explicit formula,
\[
    |w_0^{(n,\infty)}|
    =
    \left|\frac{x_2-x_1}{x_1-x_0}
    \prod_{\ell=3}^{\infty}\frac{x_\ell-x_2}{x_\ell-x_0}\right| \le \frac{x_2-x_1}{x_1-x_0} = \frac{nB-B}{B-1} \le 2n = \frac{2n}{x_0^M},
\]
since each factor in the product is at most $1$, and simply by observing that $x_0=1$. Next, for $j=1$, we similarly have
\[
    |w_1^{(n,\infty)}|
    =
    \frac{x_2-x_0}{x_1-x_0}
    \prod_{\ell=3}^{\infty}\frac{x_\ell-x_2}{x_\ell-x_1}
    \le
    \frac{x_2-x_0}{x_1-x_0}
    =
    \frac{nB-1}{B-1}
    \le 2n = \frac{2nB^M}{x_1^M},
\]
since $x_1=B$. For $j=2$, recall that $w_2^{(n,\infty)}=1$. Since $x_2=nB$, we get $|w_2^{(n,\infty)}| = 1 = \frac{(nB)^M}{x_2^M}$.
In the remainder, we let $C_M^{(0,1,2)}:=\max\{ 2n, 2nB^M, (nB)^M \}$.
We now turn to the case $j\ge 3$. Using the explicit barycentric formula and cancelling the common factor $|x_j-x_2|$, we obtain
\[
    |w_j|
    =
    \frac{(x_2-x_0)(x_2-x_1)}{(x_j-x_0)(x_j-x_1)}
    \prod_{m=3}^{j-1}\frac{x_m-x_2}{x_j-x_m}
    \prod_{\ell=j+1}^{\infty}\frac{x_\ell-x_2}{x_\ell-x_j}.
\]
We first bound the first and third factors. Since $x_j\ge x_3=(nB)^B\ge (nB)^3$, we have $x_j\ge 2$ and $x_j\ge 2B$, hence $x_j-x_0\ge \frac{x_j}{2}, x_j-x_1\ge \frac{x_j}{2}$. Therefore, $\frac{(x_2-x_0)(x_2-x_1)}{(x_j-x_0)(x_j-x_1)} \le \frac{(nB)^2}{(x_j/2)^2} = \frac{4(nB)^2}{x_j^2}$.
For the third factor, we have that,
\[
    \prod_{\ell=j+1}^{\infty}\frac{x_\ell-x_2}{x_\ell-x_j}
    \le \exp \left(\sum_{\ell=j+1}^{\infty}\frac{2x_j}{x_\ell}\right)
    \le
    \exp \left(\frac{4}{x_j^2}\right)
    \le
    \exp \left(\frac{4}{x_3^2}\right)
    <
    2.
\]
The first inequality above holds using $\ln(1+z)\le z$ for $z \ge 0$ and since, for every $\ell\ge j+1$, $x_\ell\ge x_{j+1}\ge x_j^3$, and so $x_\ell-x_j\ge x_\ell/2$. Thus, $\frac{x_\ell-x_2}{x_\ell-x_j} = 1+\frac{x_j-x_2}{x_\ell-x_j} \le 1+\frac{2x_j}{x_\ell}$. The second inequality follows since $x_{\ell+1}\ge x_\ell^3\ge 2x_\ell$ for $\ell\ge j+1$, which means that $\sum_{\ell=j+1}^{\infty}\frac{1}{x_\ell} \le \frac{2}{x_{j+1}} \le \frac{2}{x_j^3}$. The third follows by recalling that $x_j \ge x_3$, and the final bound uses $x_3=(nB)^B\ge 6^3=216$. Combining these two bounds, we obtain for every $j\ge 3$,
\[
    |w_j|
    \le
    \frac{8(nB)^2}{x_j^2}
    \prod_{m=3}^{j-1}\frac{x_m-x_2}{x_j-x_m}.
\]
If $M=1$ or $M=2$, then $L=0$, and since the middle product is at most $1$, we immediately get
\[
    |w_j|
    \le
    \frac{8(nB)^2}{x_j^2}
    \le
    \frac{8(nB)^2}{x_j^M}
    \qquad \forall j\ge 3.
\]
Thus, assume now that $M\ge 3$, so $L\ge 1$. First consider the range $3\le j\le L+2$. Since $M-2\ge 1$ and $x_j\le x_{L+2}$, the previous $M=2$ estimate yields
\[
    |w_j|
    \le
    \frac{8(nB)^2}{x_j^2}
    =
    \frac{8(nB)^2 x_j^{M-2}}{x_j^M}
    \le
    \frac{8(nB)^2 x_{L+2}^{M-2}}{x_j^M}.
\]
Now consider $j\ge L+3$. For each $m\in\{3,\dots,j-1\}$, we have $x_m>x_2$, so the factor is positive. Moreover, since $j\ge m+1$, we have $x_j\ge x_{m+1}=x_m^B$, and hence $x_j-x_m\ge x_m^B-x_m$. Also, because $m\ge 3$, we have $x_m\ge x_2>0$, and since $B\ge 3$, it follows that $x_m^B\ge x_m^2\ge 2x_m$. Therefore $x_m^B-x_m\ge x_m\ge x_m-x_2$, and so $0<\frac{x_m-x_2}{x_j-x_m}\le 1$.
Therefore
\[
\prod_{m=3}^{j-1}\frac{x_m-x_2}{x_j-x_m}
\le \prod_{m=j-L}^{j-1}\frac{x_m-x_2}{x_j-x_m}.
\]
Fix $m\in\{j-L,\dots,j-1\}$. Since $m\ge 3$, we have $x_m\ge x_2$, and since $j>m$, $x_j\ge x_{m+1}=x_m^B\ge x_m^3\ge 2x_m$. Therefore, $0< \frac{x_m-x_2}{x_j-x_m} \le \frac{x_m}{x_j/2} = \frac{2x_m}{x_j}$. Also $m\le j-1$, and hence $x_m\le x_{j-1}=x_j^{1/B}$. This means that
\[
    \frac{x_m-x_2}{x_j-x_m} \le \frac{2}{x_j^{1-1/B}}.
\]
Multiplying over $m=j-L,\dots,j-1$, we obtain $\prod_{m=3}^{j-1}\frac{x_m-x_2}{x_j-x_m} \le \left(\frac{2}{x_j^{1-1/B}}\right)^L$.
Overall, for $j\ge L+3$,
\[
    |w_j|
    \le
    \frac{8(nB)^2}{x_j^2}
    \left(\frac{2}{x_j^{1-1/B}}\right)^L
    =
    8(nB)^2 2^L x_j^{-2-L(1-1/B)} \le \frac{8(nB)^2 2^L}{x_j^M},
\]
since, by definition, $2+L(1-1/B)\ge M$. Combining the bounds for $3\le j\le L+2$ and $j\ge L+3$, we conclude that for all $j\ge 3$, $|w_j| \le \frac{C_M^{(\ge 3)}}{x_j^M}$,
where $C_M^{(\ge 3)} := 8(nB)^2\max\{ x_{L+2}^{ M-2}, 2^L \}$. Hence, the constant $C_M:=\max\{C_M^{(0,1,2)}, C_M^{(\ge 3)}\}$ satisfies
\[
    |w_j|
    \le
    \frac{C_M}{x_j^M}
    \qquad \forall j\ge 0,
\]
which proves the lemma.
\end{proof}

We also stress that the same decay bound holds uniformly in $K$:

\begin{lemma}
\label{lem:all-moment-weight-decay-uniform-K}
For every integer $M\ge 1$, the constant $C_M$ from \Cref{lem:all-moment-weight-decay} also satisfies
\[
|w_j^{(n,K)}|
\le
\frac{C_M}{x_j^M}
\qquad
\forall K\ge 0,\ \forall j\in\{0,\dots,K+2\}.
\]
\end{lemma}

\begin{proof}
For $j=0,1,2$, the same estimates as in the proof of \Cref{lem:all-moment-weight-decay} give $|w_0^{(n,K)}|\le 2n$, $|w_1^{(n,K)}|\le 2n$, $|w_2^{(n,K)}|=1$.
For $j\ge 3$, the finite-$K$ barycentric formula reads
\[
|w_j^{(n,K)}|
=
\frac{(x_2-x_0)(x_2-x_1)}{(x_j-x_0)(x_j-x_1)}
\prod_{m=3}^{j-1}\frac{x_m-x_2}{x_j-x_m}
\prod_{\ell=j+1}^{K+2}\frac{x_\ell-x_2}{x_\ell-x_j}.
\]
The proof of \Cref{lem:all-moment-weight-decay} uses only upper bounds on the factors, and the final product over $\ell$ is smaller here than in the limiting case, because all its factors are $>1$. Therefore the same argument yields $|w_j^{(n,K)}| \le \frac{C_M}{x_j^M} \qquad (j\ge 3)$.
Combining the cases $j=0,1,2$ and $j\ge 3$ proves the claim.
\end{proof}

\begin{lemma}
\label{lem:all-moments-signed-identity}
For every integer $s\ge 0$,
\[
\sum_{j\ge 0} w_jx_j^s = 0.
\]
\end{lemma}

\begin{proof}
Fix $s\ge 0$, and set $M:=s+2$. For each $K\ge 0$ and $j\ge 0$, define
\[
a_{K,j}:=
\begin{cases}
w_j^{(n,K)}x_j^s, & 0\le j\le K+2,\\
0, & j>K+2.
\end{cases}
\]
For each fixed $j$, \Cref{lem:limit-weights-exist} gives $a_{K,j}\to w_jx_j^s$ for $K\to\infty$.
By \Cref{lem:all-moment-weight-decay-uniform-K}, $|a_{K,j}| \le \frac{C_M x_j^s}{x_j^M} = \frac{C_M}{x_j^2}$ for all $K, j$.
Since $x_{j+1}\ge x_j^3\ge 2x_j$ for all $j\ge 3$, the series $\sum_{j\ge 0}x_j^{-2}$ converges. Hence dominated convergence for the counting measure on $\mathbb N$ yields
\[
\sum_{j\ge 0} w_jx_j^s
=
\lim_{K\to\infty}\sum_{j\ge 0} a_{K,j}
=
\lim_{K\to\infty}\sum_{j=0}^{K+2} w_j^{(n,K)}x_j^s.
\]
For every $K\ge \max \{0, s-1\}$, the last sum is $0$ by the defining moment-annihilation identities. Therefore the limit is $0$.
\end{proof}

This allows us to conclude that the limiting probabilities are well-defined:

\begin{corollary}
\label{cor:limit-probabilities-well-defined}
For the limiting weights, now define
\[
I_+ := \{j\ge 0:\ j \text{ even}\},
\qquad
I_- := \{j\ge 0:\ j \text{ odd}\}.
\]
Then the series $\sum_{j\ge 0}|w_j|$ converges, and
\[
W:=\sum_{j\in I_+} w_j
=
\sum_{j\in I_-}|w_j|
\in (0,\infty).
\]
Consequently,
\[
\nu^{(n,\infty)}(\{x_j\})
=:p_j
:=
\frac{w_j}{W}\mathbf 1_{\{j\in I_+\}},
\qquad
\mu^{(n,\infty)}(\{x_j\})
=:q_j
:=
\frac{|w_j|}{W}\mathbf 1_{\{j\in I_-\}}
\]
define two probability measures on $\{x_j:j\ge 0\}$. Moreover, for every $M\ge 1$,
\[
p_j+q_j
=
\frac{|w_j|}{W}
\le
\frac{C_M}{W x_j^M}
\qquad \forall j\ge 0.
\]
\end{corollary}

\begin{proof}
Taking $M=2$ in \Cref{lem:all-moment-weight-decay}, we get $\sum_{j\ge 0}|w_j| \le C_2\sum_{j\ge 0}\frac{1}{x_j^2} <\infty$.
By \Cref{lem:all-moments-signed-identity} with $s=0$, $\sum_{j\ge 0} w_j=0$.
Also $w_2^{(n,\infty)}=\lim_{K\to\infty} w_2^{(n,K)}=1>0$, so the positive part is nontrivial; hence the negative part is also nontrivial because the total signed sum is $0$. Therefore $\sum_{j\in I_+} w_j = \sum_{j\in I_-}|w_j| =:W\in(0,\infty)$.
The stated formulas for $p_j,q_j$ are nonnegative and have total mass $1$, so they define probability measures. The final bound follows immediately from \Cref{lem:all-moment-weight-decay}.
\end{proof}

We are now ready to show that the limiting measures have all moments matching:
\begin{lemma}
\label{lem:all-moments-limit}
For every integer $s\ge 0$,
\[
    \sum_{j\ge 0} x_j^s p_j < \infty,
    \qquad
    \sum_{j\ge 0} x_j^s q_j < \infty,
\]
and
\[
    \sum_{j\ge 0} x_j^s p_j
    =
    \sum_{j\ge 0} x_j^s q_j.
\]
Therefore, $\mu^{(n,\infty)}$ and $\nu^{(n,\infty)}$ match moments of every order.
\end{lemma}

\begin{proof}
Fix $s\ge 0$. By \Cref{cor:limit-probabilities-well-defined} with $M=s+2$, $x_j^s(p_j+q_j) \le \frac{C_{s+2}}{W x_j^2} \qquad (j\ge 0)$.
Since $\sum_{j\ge 0}x_j^{-2}<\infty$, both series
\[
\sum_{j\ge 0}x_j^s p_j,
\qquad
\sum_{j\ge 0}x_j^s q_j
\]
converge absolutely. Thus the $s$-th moments are finite. Moreover, by the definitions of $p_j$ and $q_j$, $p_j-q_j=\frac{w_j}{W}$ for all $j \ge 0$.
Therefore,
\[
\sum_{j\ge 0}x_j^s p_j
-
\sum_{j\ge 0}x_j^s q_j
=
\frac{1}{W}\sum_{j\ge 0} w_jx_j^s
=
0
\]
by \Cref{lem:all-moments-signed-identity}. Hence the $s$-th moments are equal for every $s\ge 0$.
\end{proof}

Having established all moments matching, we can now derive some useful bounds on the probabilities:
\begin{lemma}
\label{lem:explicit-support-bounds}
It holds that
\[
W\ge n,\qquad
p_2\ge \frac{1}{4n},\qquad
p_0\ge 1-\frac{2}{n},\qquad
q_1\ge 1-\frac{1}{n^2B^2}.
\]
\end{lemma}

\begin{proof}
We begin with the lower bound on $p_2$: for $w_0^{(n,\infty)}$, the explicit formula gives
\[
    w_0^{(n,\infty)}
    =
    \frac{x_2-x_1}{x_1-x_0}
    \prod_{\ell=3}^{\infty}\frac{x_\ell-x_2}{x_\ell-x_0} \le \frac{nB-B}{B-1} \le 2n,
\]
since each factor in the product is at most $1$. Observe that \Cref{lem:all-moment-weight-decay} applied with $M=2$ gives
\[
    |w_j| \le \frac{8(nB)^2}{x_j^2} \qquad \forall j\ge 3.
\]
Therefore, we also have
\[
    \sum_{\substack{j\ge 4\\ j\in I_+}} w_j
    \le
    \sum_{j=4}^{\infty}|w_j|
    \le
    8(nB)^2 \sum_{j=4}^{\infty}\frac{1}{x_j^2} \le \frac{16(nB)^2}{x_4^2} < 1,
\]
where the third inequality holds because, as $x_{j+1}\ge x_j^3\ge 2x_j$, the sequence $1/x_j^2$ is dominated by a geometric progression with ratio at most $1/2$, i.e., $\sum_{j=4}^{\infty}\frac{1}{x_j^2} \le \frac{2}{x_4^2}$. Hence, recalling that $w_2^{(n,\infty)}=1$,
\[
    W = w_0^{(n,\infty)}+w_2^{(n,\infty)}+\sum_{\substack{j\ge 4\\ j\in I_+}} w_j \le 2n+1+1 \le 4n \implies p_2 = \frac{w_2^{(n,\infty)}}{W} \ge \frac{1}{4n}.
\]
We proceed with the lower bound on $p_0$. Using the explicit formula for $w_1^{(n,\infty)}$, we have
\begin{align*}
    |w_1^{(n,\infty)}|
    = \frac{x_2-x_0}{x_1-x_0}
    \prod_{\ell=3}^{\infty}\frac{x_\ell-x_2}{x_\ell-x_1} = \frac{nB-1}{B-1}
    \prod_{\ell=3}^{\infty}\left(1-\frac{x_2-x_1}{x_\ell-x_1}\right).
\end{align*}
Since $x_\ell\ge x_3=(nB)^B\ge 2B$, we have $x_\ell-x_1=x_\ell-B\ge x_\ell/2$, and therefore, $0\le \frac{x_2-x_1}{x_\ell-x_1}=\frac{nB-B}{x_\ell-B}\le \frac{2nB}{x_\ell}$.
Moreover, $x_{\ell+1}\ge x_\ell^3\ge 2x_\ell$ for all $\ell\ge 3$, so $\sum_{\ell=3}^{\infty}\frac{1}{x_\ell}\le \frac{2}{x_3}$.
Hence,
\[
\sum_{\ell=3}^{\infty}\frac{x_2-x_1}{x_\ell-x_1}
\le
2nB\sum_{\ell=3}^{\infty}\frac{1}{x_\ell}
\le
\frac{4nB}{x_3}
=
\frac{4}{(nB)^{B-1}}
\le
\frac{4}{(nB)^2}.
\]
Now we first lower-bound $|w_1|$.
Set
\[
a_\ell:=\frac{x_2-x_1}{x_\ell-x_1}\qquad (\ell\ge 3).
\]
Then
\[
|w_1|
=
\frac{x_2-x_0}{x_1-x_0}
\prod_{\ell=3}^{\infty}\left(1-a_\ell\right)
=
\frac{nB-1}{B-1}\prod_{\ell=3}^{\infty}(1-a_\ell).
\]
Since $x_\ell\ge x_3=(nB)^B$ and $n,B\ge 3$, we have
$0\le a_\ell\le \frac{2nB}{x_\ell}\le \frac{2}{(nB)^{B-1}}<1.$
Hence, for every $L\ge 3$,
$\prod_{\ell=3}^{L}(1-a_\ell)\ge 1-\sum_{\ell=3}^{L} a_\ell.$
Letting $L\to\infty$ and using $\sum_{\ell=3}^\infty a_\ell\le 4/(nB)^2$, we obtain
$\prod_{\ell=3}^{\infty}(1-a_\ell)\ge 1-\frac{4}{(nB)^2}.$
Therefore
\[
|w_1|\ge \frac{nB-1}{B-1}\left(1-\frac{4}{(nB)^2}\right).
\]
Hence
\[
|w_1|-n
\ge
\frac{nB-1}{B-1}\left(1-\frac{4}{(nB)^2}\right)-n
=
\frac{n-1}{B-1}
-
\frac{4(nB-1)}{(B-1)n^2B^2}.
\]
Since $n,B\ge 3$, we have $\frac{4(nB-1)}{n^2B^2}\le \frac{4}{nB}\le \frac49 < 1 \le n-1$,
and therefore $\frac{4(nB-1)}{(B-1)n^2B^2}\le \frac{n-1}{B-1}$.
Thus $|w_1|\ge n$, and consequently $W=\sum_{j\in I_-}|w_j|\ge |w_1|\ge n$.
We now use again the bound on the sum of the remaining weights for the even indices $j \ge 4$, $1-p_0 = \frac{w_2^{(n,\infty)}+\sum_{\substack{j\ge 4\\ j\in I_+}} w_j}{W} \le \frac{1+1}{n} = \frac{2}{n}$.
We conclude with the lower bound on $q_1$: since the negative tail beyond $x_1$ is supported on odd $j\ge 3$, \Cref{lem:all-moment-weight-decay} applied with $M=2$ gives
\[
    \sum_{\substack{j\ge 3\\ j\in I_-}} |w_j|
    \le
    8(nB)^2\sum_{j=3}^{\infty}\frac{1}{x_j^2}
    \le
    \frac{16(nB)^2}{x_3^2}
    =
    \frac{16}{(nB)^{2B-2}}.
\]
Since $W\ge n, B\ge 3$ and $nB^2 \ge 18$,
\[
    1-q_1
    =
    \sum_{\substack{j\ge 3\\ j\in I_-}} q_j
    =
    \frac{1}{W}\sum_{\substack{j\ge 3\\ j\in I_-}} |w_j|
    \le
    \frac{16}{n(nB)^{2B-2}} \le \frac{16}{n^3B^4} \le \frac{1}{n^2B^2}.
\]
This proves the lemma.
\end{proof}

\subsubsection{Finishing the Proof}

\begin{proof}[Proof of \Cref{prop:key-prop-1}]
Set $\nu:=\nu^{(n,\infty)}$ and $\mu:=\mu^{(n,\infty)}$. By \Cref{cor:limit-probabilities-well-defined}, these are well-defined probability measures, and $W^+=W^-=W$. By \Cref{lem:all-moments-limit}, $\mu$ and $\nu$ agree on all moments, and all of these moments are finite.

Points $\mathrm{(i)}$-$\mathrm{(iii)}$ are exactly the support-probability bounds proved in \Cref{lem:explicit-support-bounds}, together with the trivial support facts $q_0=0$, $p_1=0$, and $q_2=0$. Finally, point $\mathrm{(iv)}$ combines the bound $W\ge n$ from \Cref{lem:explicit-support-bounds} with the case $M=2$ of \Cref{lem:all-moment-weight-decay}, which gives $|w_j|\le \frac{8(nB)^2}{x_j^2}$ for all $j\ge 0$.
This proves the proposition.
\end{proof}

\subsection{Proof of the \texorpdfstring{$N$}{N}-Way Separator}

The goal of this section is to prove \Cref{prop:key-prop-2}, of which we state an extended version below:
\begin{proposition}\label{prop:key-prop-2-ext}
    Fix integers $B,N\ge 2$ and define
    \[
        G_N(x):=\prod_{m=0}^{\infty}\sum_{u=0}^{N-1}(xB^{-m})^u
        =\sum_{j=0}^{\infty} w_j x^j.
    \]
    For \(h\in\{0,\dots,N-1\}\) and \(t\in\N_0\), let
    \[
    W_h(t):=\sum_{\substack{j\ge 0\\ j\equiv h\!\!\!\pmod N}} w_j B^{tj}.
    \]
    Then
    \[
    W_h(0)\in(0,\infty)\qquad (h=0,\dots,N-1).
    \]
    Define
    \[
    \mu_h(\{B^j\}) :=
    \begin{cases}
    \dfrac{w_j}{W_h(0)}, & j\equiv h\pmod N,\\
    0, & \text{otherwise}.
    \end{cases}
    \]
    Then the measures \((\mu_h)_{h=0}^{N-1}\) agree on all moments, all of these moments are finite, and the normalizing constants satisfy
    $W_h(0)=\sum_{\substack{j\ge 0\\ j\equiv h\!\!\!\pmod N}} w_j \in (0,\infty)$ for $h \in \{0,\dots,N-1\}$.
    Also, writing each $j$ as $j=L(N-1)+b$ for $L\in\N_0, b\in\{0,\dots,N-2\}$, and defining $e(j):=\frac{(N-1)L(L-1)}{2}+Lb$, it holds that
    \[
        B^{-e(j)}\le w_j\le 2^j B^{-e(j)}.
    \]
    Moreover, we have that
    \[
        1 \le W_0(0) = \ldots= W_{N-1}(0) < 4,
    \]
    as well as $j-e(j)\le N-1$ with equality if and only if $j\in\{N-1,N,\dots,2N-2\}$.
\end{proposition}

\subsubsection{Structural Properties of the Construction}
We next prove a sequence of structural properties of the construction that will establish the proposition:

\begin{lemma}
\label{lem:N-class-coeff-formula}
Fix integers $N\ge 2$ and $B\ge 2$, and set $q:=1/B$. For each $M\in\N_0$, define
\[
    G_{N,M}(x):=\prod_{m=0}^{M}\sum_{u=0}^{N-1}(xq^m)^u.
\]
Then:

\begin{itemize}
    \item[(i)] The infinite product
    \[
        G_N(x):=\prod_{m=0}^{\infty}\sum_{u=0}^{N-1}(xq^m)^u
    \]
    converges locally uniformly on $\C$, and hence defines an entire function.

    \item[(ii)] Writing
    \[
        G_{N,M}(x)=\sum_{j\ge 0} w_j^{(M)}x^j,
        \qquad
        G_N(x)=\sum_{j\ge 0} w_jx^j,
    \]
    one has, for every $j\ge 0$,
    \[
        w_j^{(M)}
        =
        \sum_{\substack{(c_0,\dots,c_M)\\ c_m\in\{0,\dots,N-1\}\\ \sum_{m=0}^{M} c_m=j}}
        q^{\sum_{m=0}^{M} mc_m}.
    \]

    \item[(iii)] For every fixed $j\ge 0$, the sequence $w_j^{(M)}$ is nondecreasing in $M$ and $w_j^{(M)}\uparrow w_j$ as $M\to\infty$.
    Consequently,
    \[
        w_j
        =
        \sum_{\substack{(c_m)_{m\ge 0}\\ c_m\in\{0,\dots,N-1\}\\ \sum_{m\ge 0} c_m=j}}
        q^{\sum_{m\ge 0} mc_m}
        =
        \sum_{\substack{(c_m)_{m\ge 0}\\ c_m\in\{0,\dots,N-1\}\\ \sum_{m\ge 0} c_m=j}}
        B^{-\sum_{m\ge 0} mc_m}.
    \]
\end{itemize}
\end{lemma}

\begin{proof}
For each $m\ge 0$, set $F_m(x):=\sum_{u=0}^{N-1}(xq^m)^u$.
Fix $R>0$. For $|x|\le R$,
\[
    |F_m(x)-1|
    \le
    \sum_{u=1}^{N-1}|x|^u q^{um}
    \le
    \sum_{u=1}^{N-1}R^u q^{um}.
\]
Hence
\[
    \sum_{m=0}^{\infty}\sup_{|x|\le R}|F_m(x)-1|
    \le
    \sum_{u=1}^{N-1}R^u\sum_{m=0}^{\infty}q^{um}
    <
    \infty.
\]
Therefore the infinite product $\prod_{m=0}^{\infty}F_m(x)$ converges locally uniformly on $\C$ to a holomorphic function, proving (i).

Now fix $M\in\N_0$. Since $G_{N,M}$ is a finite product of polynomials, expanding it gives
\[
    G_{N,M}(x)
    =
    \sum_{(c_0,\dots,c_M)\in\{0,\dots,N-1\}^{M+1}}
    \prod_{m=0}^{M}(xq^m)^{c_m}.
\]
Collecting powers of $x$, we obtain
\[
    G_{N,M}(x)
    =
    \sum_{(c_0,\dots,c_M)\in\{0,\dots,N-1\}^{M+1}}
    x^{\sum_{m=0}^{M}c_m}q^{\sum_{m=0}^{M}mc_m},
\]
so the coefficient of $x^j$ is exactly
\[
    w_j^{(M)}
    =
    \sum_{\substack{(c_0,\dots,c_M)\\ c_m\in\{0,\dots,N-1\}\\ \sum_{m=0}^{M} c_m=j}}
    q^{\sum_{m=0}^{M} mc_m}.
\]
This proves (ii). Next, because $G_{N,M+1}(x)=G_{N,M}(x)\,F_{M+1}(x)$
and all coefficients of $F_{M+1}(x)$ are nonnegative, every coefficient of $G_{N,M+1}$ is at least the corresponding coefficient of $G_{N,M}$. Thus, for each fixed $j$, $w_j^{(M+1)}\ge w_j^{(M)}$.
So $(w_j^{(M)})_{M\ge 0}$ is nondecreasing.

Since $G_{N,M}\to G_N$ locally uniformly on $\C$, the Taylor coefficients converge coefficientwise. For example, fixing any $r>0$, Cauchy's coefficient formula gives
\[
    w_j^{(M)}
    =
    \frac{1}{2\pi i}\int_{|z|=r}\frac{G_{N,M}(z)}{z^{j+1}}\,dz
    \longrightarrow
    \frac{1}{2\pi i}\int_{|z|=r}\frac{G_N(z)}{z^{j+1}}\,dz
    =
    w_j.
\]
Hence $w_j^{(M)}\uparrow w_j$. Now, for fixed $j\ge 0$, let
\[
\mathcal E_{j,M}
:=
\left\{(c_m)_{m\ge 0} :
c_m\in\{0,\dots,N-1\},\
\sum_{m\ge 0} c_m=j,\
c_m=0\ \forall m>M
\right\},
\]
and let
$
\mathcal E_j:=\bigcup_{M\ge 0}\mathcal E_{j,M}.
$
Then part (ii) gives
$
w_j^{(M)}=\sum_{(c_m)\in \mathcal E_{j,M}} q^{\sum_{m\ge 0} mc_m}.
$
Since $\mathcal E_{j,M}\uparrow \mathcal E_j$ and all summands are nonnegative, monotone convergence yields
$
\lim_{M\to\infty} w_j^{(M)}
=
\sum_{(c_m)\in \mathcal E_j} q^{\sum_{m\ge 0} mc_m}.
$
Because $w_j^{(M)}\to w_j$, we obtain
\[
w_j
=
\sum_{\substack{(c_m)_{m\ge 0}\\ c_m\in\{0,\dots,N-1\}\\ \sum_{m\ge 0} c_m=j}}
q^{\sum_{m\ge 0} mc_m}.
\]
This is well defined because every admissible sequence $(c_m)_{m\ge 0}$ has finite support once $\sum_{m\ge 0}c_m=j<\infty$. Recalling $q=1/B$ gives the last identity.
\end{proof}

With these coefficients, we now show that the residue-class sums $W_h(t)$ are equal for every moment order $t\in\N_0$.
\begin{lemma}
\label{lem:N-class-moment-matching}
For every $t\in\N_0$ and every $h\in\{0,\dots,N-1\}$, let the residue-class sums be defined as $W_h(t):=\sum_{\substack{j\ge 0\\ j\equiv h\pmod N}} w_j B^{tj}$. Then, for all $h \in \{0, \ldots, N-1\}$, it holds that:
$W_h(t) = \frac{G_N(B^t)}{N}$.
Moreover, letting $W:= W_0(0) = \ldots= W_{N-1}(0)$, we have that $1 \le W < 4$.
\end{lemma}

\begin{proof}
Let $\omega:=e^{2\pi \mathrm{i}/N}$ be a primitive $N$-th root of unity. Fix $t\in\N_0$ and $\ell\in\{1,\dots,N-1\}$. In the defining product for $G_N(\omega^\ell B^t)$, the factor with index $m=t$ is $\sum_{u=0}^{N-1}(\omega^\ell)^u=0,$
hence $G_N(\omega^\ell B^t)=0$. Since
$
\sum_{j=0}^{\infty}\left|w_j(\omega^\ell B^t)^j\right|
=
\sum_{j=0}^{\infty} w_j B^{tj}
=
G_N(B^t)
<
\infty,
$
the series is absolutely convergent, so we may regroup it by residue classes modulo $N$. Therefore
$$
0
=
\sum_{j=0}^{\infty} w_j(\omega^\ell B^t)^j
=
\sum_{h=0}^{N-1}\ \sum_{\substack{j\ge 0\\ j\equiv h\!\!\!\pmod N}} w_j(\omega^\ell B^t)^j
=
\sum_{h=0}^{N-1}\omega^{\ell h}\sum_{\substack{j\ge 0\\ j\equiv h\!\!\!\pmod N}} w_j B^{tj}
=
\sum_{h=0}^{N-1}\omega^{\ell h}W_h(t).
$$
Note also that, by definition of $W_h(t)$'s, we have that 
\[
    \sum_{h=0}^{N-1} W_h(t) = \sum_{h=0}^{N-1}\sum_{\substack{j\ge 0\\ j\equiv h\pmod N}} w_j B^{tj} = \sum_{j=0}^{\infty}w_j B^{tj} = G_N(B^t).
\]
We can summarize, for every fixed $t \in \N_0$, all of the above in a system of $N$ linear equations (indexed by $\ell \in \{0,\dots,N-1\}$) in $N$ unknowns (indexed by $h \in \{0,\dots,N-1\}$). Namely, consider the matrix $\mathbf \Omega \in \mathbb C^{N \times N}$ that consists of $\mathbf \Omega_{\ell h} = \omega^{\ell h}$'s as entries and the vector $\mathbf W(t) \in \R^N_+$ consists of $W_h(t)$'s as entries. Then,
\[
    \begin{pmatrix}
    1 & 1 & 1 & \cdots & 1\\ 
    1 & \omega & \omega^2 & \cdots & \omega^{N-1}\\
    1 & \omega^2 & \omega^4 & \cdots & \omega^{2(N-1)}\\ \vdots & \vdots & \vdots & & \vdots\\
    1 & \omega^{N-1} & \omega^{2(N-1)} & \cdots & \omega^{(N-1)^2}
    \end{pmatrix} \begin{pmatrix}
    W_0(t) \\ 
    W_1(t)\\
    W_2(t)\\
    \vdots\\
    W_{N-1}(t)
    \end{pmatrix} = \begin{pmatrix}
    G_N(B^t) \\ 
    0\\
    0\\
    \vdots\\
    0
    \end{pmatrix}.
\]
 We stress that $\mathbf \Omega$ is the Vandermonde matrix of the $N$-th roots of unity and we know that this is invertible because $1 \neq \omega \neq \cdots \neq \omega^{N-1}$, since $\omega$ is a primitive $N$-th root of unity. Thus, the solution to the above must be unique. Suppose it takes the form $\mathbf W(t) = \gamma \mathbf 1$ for some $\gamma \in \R$, then it must satisfy $\gamma \mathbf 1^\top \mathbf 1 = G_N(B^t)$, i.e., $\gamma = \frac{G_N(B^t)}{N}$, which, for all $\ell \in \{1, \ldots, N-1\}$, also satisfies $\gamma \sum_{h=0}^{N-1} \omega^{\ell h} = 0$, since $\sum_{h=0}^{N-1} \omega^{\ell h} = \frac{1 - \omega^{\ell N}}{1 - \omega^\ell} = 0$. Hence, $W_0(t)=W_1(t)=\cdots=W_{N-1}(t) = \frac{G_N(B^t)}{N}$.

 To conclude the proof, observe that
    \[
    G_N(1)
    =
    N\prod_{m=1}^{\infty}\sum_{u=0}^{N-1} q^{um}
    =
    N\prod_{m=1}^{\infty}\frac{1-q^{Nm}}{1-q^m},
    \]
and so
\[
    W=\frac{G_N(1)}{N} =\prod_{m=1}^{\infty}\frac{1-q^{Nm}}{1-q^m} \le \prod_{m=1}^{\infty}\frac{1}{1-2^{-m}} \le\frac{32}{9}<4,
\]
For the last bound, $\prod_{m=3}^{\infty}(1-2^{-m})\ge1-\sum_{m=3}^{\infty}2^{-m}=3/4$; including the first two factors gives $\prod_{m=1}^{\infty}(1-2^{-m})\ge9/32$. Finally, $w_0=1$ and $W=W_0(0)\ge w_0$, so $W\ge1$.
\end{proof}

Based on the above result, using the definition $W:=W_0(0)=\cdots=W_{N-1}(0)$ and, for each $h\in\{0,\dots,N-1\}$, the distributions descriptions simplify to 
\[
    \mu_h(\{B^j\}):= \begin{cases}
        \dfrac{w_j}{W_h(0)} & j\equiv h\pmod N\\
        0 & \text{otherwise}
    \end{cases} = 
    \begin{cases}
        \dfrac{w_j}{W}, & j\equiv h\pmod N\\
        0, & \text{otherwise}
    \end{cases}.
\]
We can now observe that these distributions share all moments, which are additionally finite:
\begin{lemma}
\label{lem:N-class-all-moments}
For every $h,h'\in\{0,\dots,N-1\}$ and every $t\in\N_0$,
$\int x^t\, d\mu_h(x)=\int x^t\, d\mu_{h'}(x)<\infty$.
\end{lemma}

\begin{proof}
By \Cref{lem:N-class-moment-matching},
\[
    \int x^t d\mu_h(x)
    =
    \frac{1}{W}\sum_{\substack{j\ge 0\\ j\equiv h\pmod N}} w_j B^{tj}
    =
    \frac{W_h(t)}{W}
    =
    \frac{W_{h'}(t)}{W}
    =
    \int x^t d\mu_{h'}(x).
\]
Since $G_N$ is entire and the coefficients are nonnegative, the series $\sum_{j=0}^{\infty} w_j B^{tj}=G_N(B^t)$ converges for every $t\in\N_0$, so all moments are finite.
\end{proof}

\subsubsection{Estimating Decay Rates of Limiting Weights}
We can next provide bounds on the coefficients $w_j$ that define the distributions above. To this end, it is useful to write for the remainder of the proof each $j\in\N_0$ as $j=L(N-1)+b$ for $L\in\N_0, b\in\{0,\dots,N-2\}$, and to define $e(j):=\frac{(N-1)L(L-1)}{2}+Lb$.

\begin{lemma}
\label{lem:N-class-coeff-bounds}
Assume $B\ge 2$. Then for every $j\in\N_0$, $B^{-e(j)}\le w_j\le 2^j B^{-e(j)}$.
Moreover, $j-e(j)\le N-1$ with equality if and only if
$j\in\{N-1,N,\dots,2N-2\}$.
\end{lemma}

\begin{proof}
By \Cref{lem:N-class-coeff-formula}, for every $j\in\N_0$,
\[
    w_j
    =
    \sum_{\substack{(c_m)_{m\ge 0}\\ c_m\in\{0,\dots,N-1\}\\ \sum_{m\ge 0} c_m=j}}
    B^{-\sum_{m\ge 0} mc_m}.
\]
We use this formula to prove both bounds.
Write $j=L(N-1)+b$ for $L\in\N_0, b\in\{0,\dots,N-2\}$,
and recall that $e(j):=\frac{(N-1)L(L-1)}{2}+Lb$.

We begin with the lower bound. Consider the admissible sequence $c_0=\cdots=c_{L-1}=N-1,\qquad c_L=b,\qquad c_m=0 \ \text{for all } m\ge L+1$.
Its total mass is $\sum_{m\ge 0} c_m = L(N-1)+b = j$,
so it contributes to the above sum for $w_j$. The corresponding exponent is
\[
    \sum_{m\ge 0} mc_m
    =
    \sum_{m=0}^{L-1} m(N-1)+Lb
    =
    (N-1)\frac{L(L-1)}{2}+Lb
    =
    e(j).
\]
Since all terms in the coefficient formula are nonnegative, this single contribution yields $w_j\ge B^{-e(j)}$.

We next prove the upper bound. For each admissible sequence $(c_m)_{m\ge 0}$, let $0\le t_1\le t_2\le \cdots \le t_j$ be the nondecreasing sequence obtained by listing each index $m$ exactly $c_m$ times. Then $\sum_{m\ge 0} mc_m=\sum_{i=1}^{j} t_i$.
Moreover, because each $c_m\le N-1$, no integer can occur more than $N-1$ times in the sequence $(t_i)_{i=1}^j$. It follows that $t_i\ge \left\lfloor\frac{i-1}{N-1}\right\rfloor$ for every $i=1,\dots,j$.
Therefore we may write $t_i=\left\lfloor\frac{i-1}{N-1}\right\rfloor+\delta_i$, for $\delta_i\in\N_0$.
Summing over $i$ gives
\[
    \sum_{i=1}^{j} t_i
    =
    \sum_{i=1}^{j}\left\lfloor\frac{i-1}{N-1}\right\rfloor
    +
    \sum_{i=1}^{j}\delta_i.
\]
Now
\[
    \sum_{i=1}^{j}\left\lfloor\frac{i-1}{N-1}\right\rfloor
    =
    \sum_{r=0}^{L-1} r(N-1)+Lb
    =
    \frac{(N-1)L(L-1)}{2}+Lb
    =
    e(j),
\]
because among the integers $0,1,\dots,j-1=L(N-1)+b-1$,
each value $r\in\{0,\dots,L-1\}$ occurs exactly $N-1$ times in the quotient
$\lfloor(i-1)/(N-1)\rfloor$, and the value $L$ occurs exactly $b$ times.
Hence $\sum_{m\ge 0} mc_m = \sum_{i=1}^{j} t_i = e(j)+\sum_{i=1}^{j}\delta_i$.

Moreover, the admissible sequence $(c_m)_{m\ge 0}$ uniquely determines the nondecreasing list
$0\le t_1\le \cdots \le t_j,$
and hence uniquely determines the tuple
$\delta_i:=t_i-\left\lfloor\frac{i-1}{N-1}\right\rfloor \in \N_0
\quad (i=1,\dots,j).$
Thus, the map from admissible sequences to $\delta$-tuples is injective. Therefore,
\[
w_j
=
\sum_{\text{admissible }(c_m)}
B^{-e(j)-\sum_{i=1}^j\delta_i}
\le
B^{-e(j)}
\sum_{\delta_1,\dots,\delta_j\ge 0}
B^{-\sum_{i=1}^{j}\delta_i}.
\]
The variables now separate, so
\[
    \sum_{\delta_1,\dots,\delta_j\ge 0}
    B^{-\sum_{i=1}^{j}\delta_i}
    =
    \left(\sum_{h=0}^{\infty} B^{-h}\right)^j
    =
    \left(\frac{1}{1-B^{-1}}\right)^j
    \le 2^j,
\]
since $B\ge 2$. Therefore $w_j\le 2^j B^{-e(j)}$.
This proves $B^{-e(j)}\le w_j\le 2^j B^{-e(j)}$.
It remains to estimate $j-e(j)$. Using $j=L(N-1)+b$, we compute $j-e(j) = L(N-1)+b-\frac{(N-1)L(L-1)}{2}-Lb$ and distinguish cases:

\begin{enumerate}
    \item If $L=0$, then $j=b$ and $j-e(j)=b\le N-2$.
    \item If $L=1$, then $j-e(j)=(N-1)+b-b=N-1$. Since in this case $j=(N-1)+b$ with $b\in\{0,\dots,N-2\}$, this is exactly the range $j\in\{N-1,N,\dots,2N-3\}$.
    \item If $L=2$, then $j-e(j)=2(N-1)+b-(N-1)-2b=N-1-b$. Thus, \(j-e(j)=N-1\) if and only if \(b=0\), i.e., if and only if $j=2N-2$, and otherwise \(j-e(j)\le N-2\).
    \item Finally, if $L\ge 3$, then $j-e(j) = \frac{(3L-L^2)(N-1)}{2}+(1-L)b$. Here \(3L-L^2\le 0\) and \(1-L\le 0\), so $j-e(j)\le 0\le N-2$.
\end{enumerate}
Combining all cases, we conclude that $j-e(j)\le N-1$, with equality if and only if $j\in\{N-1,N,\dots,2N-2\}$.
This completes the proof.
\end{proof}

\subsubsection{Finishing the Proof}

We now conclude the proof of \Cref{prop:key-prop-2}:

\begin{proof}[Proof of \Cref{prop:key-prop-2}]
By \Cref{lem:N-class-coeff-formula}, all coefficients $w_j$ are nonnegative. Moreover, by \Cref{lem:N-class-moment-matching},
\[
W_0(0)=\cdots=W_{N-1}(0)=:W,
\qquad 1\le W<4.
\]
Hence, for each $h\in\{0,\dots,N-1\}$, the formula
\[
\mu_h(\{B^j\}):=
\begin{cases}
\dfrac{w_j}{W}, & j\equiv h\pmod N,\\
0, & \text{otherwise},
\end{cases}
\]
defines a probability measure supported on $\{B^j: j\equiv h\pmod N\}$.

Next, by \Cref{lem:N-class-all-moments}, the distributions $\mu_h$ agree on all moments, and all these moments are finite. \Cref{lem:N-class-coeff-bounds} also shows that, for every $j\in\N_0$, $B^{-e(j)}\le w_j\le 2^j B^{-e(j)}$, and also that $j-e(j)\le N-1$ with equality if and only if $j\in\{N-1,N,\dots,2N-2\}$. This proves the proposition.
\end{proof}

%% file: sections/discussion.tex
\section{Concluding Remarks}
\label{sec:discussion}

Our results all stem from the same gap between moment information and the local features of a distribution that matter for optimization. Distributions with the same full moment sequence can have quantiles separated by arbitrarily large factors, with consequences for Value-at-Risk, lower-tail average quantiles, chance-constrained capacity, and newsvendor decisions. They can also have disjoint ranges of nearly optimal posted prices. For secretary success probability, the $N$-way construction concentrates distributions near different first atoms; uniform mixtures of these laws and the Ramsey argument yield the classical finite-horizon value $s_m$, even for i.i.d.\ observations. For prophet inequalities, the binary construction creates the order-statistic gap behind the fixed-$r$ barriers. Moments determine expectations of polynomials, but they need not reveal the local tail or threshold behavior on which these decisions depend.

\section{Acknowledgements}
A substantial portion of this work was completed by the authors in early 2026. In preparing the first version, we used generative AI tools for limited purposes, primarily language editing and presentation. ChatGPT 5.6 Pro/Ultra also assisted with reviewing the exposition and identifying connections to the classical literature. For the second version, GPT-6 Astra Ultra assisted with mathematical proof development for the i.i.d.\ extension of the secretary success-probability result.

All results of this paper have been formalized in Lean 4.32.2 using mathlib. ChatGPT 5.6 Pro/Ultra assisted with reviewing the formalization for the first version. For the second version, GPT-6 Astra Ultra helped update it to cover the i.i.d.\ extension and improve its overall structure. The development is available in the \href{https://github.com/GO-EPFL/moment-ambiguity-formalization}{accompanying GitHub repository}, which includes a guide linking the paper's results to the corresponding Lean declarations. The authors take full responsibility for the paper's content.